\documentclass[envcountsect,envcountsame]{llncs}

\usepackage[cmex10]{amsmath}
\usepackage{amssymb, amsfonts}
\usepackage{calrsfs}
\usepackage{graphicx}
\usepackage{fixltx2e}
\usepackage{paralist}

\newcommand{\MAC}{\textnormal{MAC}}

\renewcommand{\a}{\mathbf{a}}
\renewcommand{\b}{\mathbf{b}}

\renewcommand{\t}{\mathbf{t}}
\renewcommand{\u}{\mathbf{u}}

\newcommand{\x}{\mathbf{x}}
\newcommand{\y}{\mathbf{y}}
\newcommand{\z}{\mathbf{z}}

\newcommand{\A}{\mathcal{A}}
\newcommand{\B}{\mathcal{B}}
\newcommand{\C}{\mathcal{C}}

\newcommand{\F}{\mathcal{F}}
\newcommand{\G}{\mathcal{G}}

\newcommand{\J}{\mathcal{J}}
\newcommand{\K}{\mathcal{K}}

\renewcommand{\P}{\mathcal{P}}
\newcommand{\R}{\mathcal{R}}
\newcommand{\T}{\mathcal{T}}
\newcommand{\V}{\mathcal{V}}

\newcommand{\X}{\mathcal{X}}
\newcommand{\Y}{\mathcal{Y}}
\newcommand{\Z}{\mathcal{Z}}
\newcommand{\U}{\mathcal{U}}

\newcommand{\EE}{\mathbb{E}}
\newcommand{\PP}{\mathbb{P}}
\newcommand{\RR}{\mathbb{R}}

\newcommand{\eps}{\varepsilon}
\renewcommand{\phi}{\varphi}

\newcommand{\eins}{\textsubscript{1}}
\newcommand{\zwei}{\textsubscript{2}}
\newcommand{\nue}{$_\nu$}

\newcommand{\0}{^{(0)}}
\newcommand{\1}{^{(1)}}
\newcommand{\2}{^{(2)}}
\newcommand{\3}{^{(3)}}

\newcommand{\hq}{^{(q)}}
\newcommand{\hr}{^{(r)}}
\newcommand{\hqr}{^{(q,r)}}
\newcommand{\halpha}{^{(\alpha)}}
\newcommand{\zhalpha}{^{(2,\alpha)}}

\author{Moritz Wiese and Holger Boche}
\institute{Technische Universit\"at M\"unchen}
\title{Strong Secrecy for Multiple Access Channels}

\begin{document}

\maketitle

\begin{center}\small\textit{
  Dedicated to the memory of Rudolf Ahlswede}
\end{center}

\begin{abstract}
  We show strongly secret achievable rate regions for two different wiretap multiple-access channel coding problems. In the first problem, each encoder has a private message and both together have a common message to transmit. The encoders have entropy-limited access to common randomness. If no common randomness is available, then the achievable region derived here does not allow for the secret transmission of a common message. The second coding problem assumes that the encoders do not have a common message nor access to common randomness. However, they may have a conferencing link over which they may iteratively exchange rate-limited information. This can be used to form a common message and common randomness to reduce the second coding problem to the first one. We give the example of a channel where the achievable region equals zero without conferencing or common randomness and where conferencing establishes the possibility of secret message transmission. Both coding problems describe practically relevant 
networks which need to be secured against eavesdropping attacks.
\end{abstract}

\section*{Contents}

\noindent\textbf{\ref{sect:intro} Introduction\hfill\pageref{sect:intro}}\medskip

\noindent\textbf{\ref{sect:wiretapMAC} The Wiretap Multiple-Access Channel\hfill\pageref{sect:wiretapMAC}}\\
\indent\ref{subsect:commmess} With Common Message\dotfill\pageref{subsect:commmess}\\
\indent\ref{subsect:confenc} With Conferencing Encoders\dotfill\pageref{subsect:confenc}\medskip

\noindent\textbf{\ref{sect:thethms} Coding Theorems\hfill\pageref{sect:thethms}}\\
\indent\ref{subsect:commmessthm} For the Wiretap MAC with Common Message\dotfill\pageref{subsect:commmessthm}\\
\indent\ref{subsect:confencthm} For the Wiretap MAC with Conferencing Encoders\dotfill\pageref{subsect:confencthm}\medskip

\noindent\textbf{\ref{proofcomm} Proof of Theorem \ref{thmcomm}\hfill\pageref{proofcomm}}\\
\indent\ref{subsect:elrateregs} Elementary Rate Regions\dotfill\pageref{subsect:elrateregs}\\
\indent\ref{subsect:provesec} How to Prove Secrecy\dotfill\pageref{subsect:provesec}\\
\indent\ref{sect:setup} Probabilistic Bounds for Secrecy\dotfill\pageref{sect:setup}\\
\indent\ref{subsect:randcodcommmess} Random Coding for the Non-Wiretap MAC with Common Message\dotfill\pageref{subsect:randcodcommmess}\\
\indent\ref{MAC} Coding\dotfill\pageref{MAC}\\
\indent\ref{subsect:concl} Concluding Steps\dotfill\pageref{subsect:concl}\medskip

\noindent\textbf{\ref{sect:proofconf} Proof of Theorem \ref{thmconf}\hfill\pageref{sect:proofconf}}\\
\indent\ref{subsect:elratreg} Elementary Rate Regions\dotfill\pageref{subsect:elratreg}\\
\indent\ref{subsect:confcod} Coding\dotfill\pageref{subsect:confcod}\medskip

\noindent\textbf{\ref{sect:discussion} Discussion\hfill\pageref{sect:discussion}}\\
\indent\ref{subsect:confsectrans} Conferencing and Secret Transmission\dotfill\pageref{subsect:confsectrans}\\
\indent\ref{subsect:timesharing} Necessity of Time-Sharing in Random Coding\dotfill\pageref{subsect:timesharing}\medskip

\noindent\textbf{\ref{sect:lemmaproof1} Proof of Lemma \ref{gemconc}\hfill\pageref{sect:lemmaproof1}}\medskip

\noindent\textbf{\ref{sect:lemmaproof2} Proof of Lemma \ref{unionconv2}\hfill\pageref{sect:lemmaproof2}}

\section{Introduction}\label{sect:intro}

The wiretap Multiple-Access Channel (MAC) combines two areas where Rudolf Ahlswede has made major contributions. In the area of multi-user information theory, he \cite{AMAC} and Liao \cite{Liao} independently gave one of the first complete characterizations of the capacity region of a multi-user channel -- the MAC with one message per sender. Later, Dueck \cite{DMACConv} proved the strong converse for the MAC and Ahlswede \cite{AMACConv} gave an elementary proof immediately afterwards. Slepian and Wolf generalized the results from \cite{AMAC} and \cite{Liao} to the case where the senders additionally have a common message \cite{SW}. Willems used Slepian and Wolf's result to derive the capacity region of the MAC with conferencing encoders. This is a MAC without common message, but the encoders can exchange rate-limited information about their messages in an interactive conferencing protocol \cite{Wi1,Wi2}. The results of Slepian and Wolf as well as Willems' result were only recently generalized to general 
compound MACs with partial channel state information in \cite{WBBJ11}, arbitrarily varying MACs with conferencing encoders were treated in \cite{WB11}. The latter paper made substantial use of techniques developed by Ahlswede for single-sender arbitrarily varying channels in \cite{A1,A2,A3} and also of his and Cai's contribution to arbitrarily varying MACs \cite{AC}. 

The other area of Ahlswede's interest which plays a role in this paper is secrecy and common randomness. Among other problems, he considered together with Csisz\'ar in \cite{ACs1,ACs2} how a secret key can be shared at distant terminals in the presence of an eavesdropper. Work on secret key sharing aided by public communication goes back to Maurer \cite{Ma}. The first paper which exploits the statistics of a discrete memoryless channel to establish secret communication is due to Wyner \cite{Wy}. He considers the wiretap channel, the simplest model of a communication scenario where secrecy is relevant: a sender would like to transmit a message to a receiver over a discrete memoryless channel and transmission is overheard by a second receiver who should be kept ignorant of the message. It was noted by Wyner that a secret key shared at both legitimate terminals is not necessary to establish secret transmission -- if the channel statistics are taken into consideration, it is sufficient that the sender randomizes 
his inputs in order to secure transmission.

Since Wyner discovered this fact, information-theoretic secrecy for message transmission without a key shared between sender and legitimate receiver has been generalized in various directions. The first paper on multi-user information-theoretic security is due to Csisz\'ar and K\"orner \cite{CKBCC}. Here, the second receiver only is a partial eavesdropper: there is a common message intended for both receivers, but as in the original wiretap channel, an additional private message intended for the first receiver must be kept secret from the second. We come to multiple-access models below. An overview over the area is given in \cite{LPS}. 

The original secrecy criterion used in \cite{Wy} and \cite{CKBCC} and in most of the subsequent work until today has become known as the ``weak secrecy criterion''. Given a code, it measures the mutual information normalized by the code blocklength between the randomly chosen message and the eavesdropper's output corresponding to the application of the code and transmission over the channel. Maurer introduced the ``strong secrecy criterion'' in \cite{Maurer94a} by omitting the normalization. The advantage of this criterion was revealed in \cite{BBS}: it can be given an operational meaning, i.e. one can specify the attacks it withstands. It is possible to show that if transmission obeys the strong secrecy criterion, then the eavesdropper's average error tends to one for any decoder it might apply. Translated into practical secrecy schemes, this means that no matter how large the computing power of a possible eavesdropper might be, it will not succeed in breaking the security of this scheme. For the weak 
criterion, there are still only heuristic argumentations as to why it should be secret. Further secrecy metrics are presented in \cite{BL}, but without giving them an operational meaning, strong secrecy remains the strongest of these metrics. To our knowledge, there are three different approaches to establishing strong secrecy in a wiretap channel so far \cite{Ma,Cs99,De05}. In fact, the last of these approaches also applies to classical-quantum wiretap channels \cite{De05} and also was used to give an achievable rate for the classical compound wiretap channel \cite{BBS}.

There exist many MAC models where secrecy is an issue. This may even be the case when there is no eavesdropper, as each encoder might have access to noisy observations of the other sender's codeword but wants to protect its own message from decoding at the other sender \cite{LMYS,LP08,EU}. The case where the encoders have access to generalized feedback but only keep their messages secret from an external eavesdropper is considered in \cite{TLSP}. In the cognitive MAC, only one encoder has a private message, and together, the encoders have a common message. There are again two cases: In the case without an eavesdropper, the encoder without a private message has access to the codeword sent by the other encoder through a noisy channel and must be kept ignorant of the other encoder's private message \cite{LLP}. In \cite{SY}, the cognitive MAC without feedback was investigated where the messages must be kept secret from an eavesdropper and the encoders have unrestricted access to common randomness. All of these 
papers use the weak secrecy criterion.

The first part of this article generalizes and strengthens the achievability result from \cite{EUMAC} where multi-letter characterizations of an achievable region and of an outer bound on the capacity region of a MAC without common message and with an external eavesdropper under the weak secrecy criterion are given. The channel needs to satisfy certain relatively strong conditions for the bounds to work. Extensions to the Gaussian case can be found in \cite{EUMAC,TY1,HY}. 

We consider two senders Alice\eins\ and Alice\zwei. Each has a private message and together they have a common message. This message triple must be transmitted to Bob over a discrete memoryless MAC in such a way that Eve who obtains a version of the sent codewords through another discrete memoryless MAC cannot decode the messages. We apply the strong secrecy criterion. In order to find a code which satisfies this criterion, we use Devetak's approach \cite{De05}, which in the quantum case builds on the Ahlswede-Winter lemma \cite{AW} and classically on a Chernoff bound. It is similar to the approach taken in \cite{CWY}. As the senders have a common message and as the second part of the paper deals with the wiretap MAC with conferencing encoders, we assume that the encoders have access to a restricted amount of common randomness. Common randomness for encoding has so far only been used in \cite{SY}, but without setting any limitations on its amount. Note that this use of common randomness in order to establish 
secrecy differs from the use made in \cite{ACs1,ACs2}. We only obtain an achievable region. In this achievable region it is not possible to transmit a common message if no common randomness is available. Further it is notable that we use random coding and have to apply time-sharing before derandomizing. 

The wiretap MAC with common message and common randomness is also needed in the second part of this paper about the wiretap MAC with conferencing encoders. Conferencing was introduced by Willems in \cite{Wi1,Wi2} and is an iterative protocol for the senders of a MAC to exchange information about their messages. One assumes that the amount of information that is exchanged is rate-limited because otherwise one would obtain a single-encoder wiretap MAC. Willems already used the coding theorem for the MAC with common message to deduce an achievable region for the conferencing MAC. The same can be done for the wiretap MAC with conferencing encoders. More precisely, aside from the senders' private messages, there are no further messages to be transmitted, and no common randomness is available. However, conferencing is used to produce both a common message and common randomness, which allows the reduction. A consequence of the fact that no common message can be transmitted by the wiretap MAC with common message if 
there is no common randomness is that one has to use conferencing to establish some common randomness if this is supposed to enlarge the achievable region compared to what would be achievable without conferencing. Again, this consequence presumes that the achievable region equals the capacity region even though we cannot prove this.

Information-theoretic security has far-reaching practical consequences. As digital communication replaces more and more of the classical paper-based ways of communication even for the transmission of sensible data, the problem of securing these data becomes increasingly important. Information-theoretic secrecy provides an alternative to the traditional cryptographic approach which bases on the assumption of limited computing power. However, as information-theoretic security uses the imperfections of the channels to secure data, its models must be sufficiently complex to describe realistic scenarios. Our article shows how encoder cooperation can be utilized to secure data. The cooperation of base stations in mobile networks is included in future wireless network standards, and our work can be seen as a contribution to the theoretical analysis of how it fares when it comes to security. But already Csisz\'ar and K\"orner's paper on the broadcast channel with confidential messages shows how messages with 
different secrecy requirements can be combined in one transmission. A more recent example which also applies the strong secrecy criterion is given in \cite{WWB}.

\subsubsection{Organization of the paper:} The next section introduces the general model of a wiretap MAC and also presents the Willems conferencing protocol. Section \ref{sect:thethms} contains the two achievability theorems for the wiretap MAC with common message and the wiretap MAC with conferencing encoders. 

The common message theorem is treated in the rather long Section \ref{proofcomm}. First, the regions we claim to be achievable are decomposed into regions whose achievability can be shown more easily. Following Devetak, it is shown that it is sufficient to make Eve's output probability given a message triple almost independent of this triple in terms of variation distance. Then, in the mathematical core of the paper, we derive lower bounds on the randomness necessary to achieve strong secrecy using probabilistic concentration results. Here we also follow Devetak. Having derived these bounds, we finally find a realization of the random codes which defines a good wiretap code.

Section \ref{sect:proofconf} gives the proof of the achievability theorem for the wiretap MAC with conferencing encoders. We again have to decompose the claimed regions into regions whose achievability can be shown more easily. Then we can reduce the problem of achieving a certain rate pair with conferencing to the problem of achieving a certain rate triple by the wiretap MAC with common message more or less in the same way as done by Willems in the non-wiretap situation. Finally, Section \ref{sect:discussion} shows that conferencing may help in situations where no secret transmission is possible without and that for our approach it is necessary to do the time-sharing within the random coding.

\subsubsection{Notation:} For sets $\{1,\ldots,M\}$, where $M$ is a positive integer, we use the combinatorial shorthand $[M]$. For a real number $x$ we define $[x]_+:=\max\{x,0\}$.

For any set $\X$ and subset $A\subset\X$ we write $A^c:=\X\setminus A$. We let $1_A:\X\rightarrow\{0,1\}$ be the indicator function of $A$ which takes on the value 1 at $x\in\X$ if and only if $x\in A$. Given a probability space $(\Omega,\A,\PP)$ we write $\EE$ for the expectation corresponding to $\PP$ and for $A\in\A$ and a real-valued random variable $X$ we write $\EE[X;A]:=E[X1_A]$. 

The space of probability distributions on the finite set $\X$ is denoted by $\P(\X)$. In particular, it contains for every $x\in\X$ the probability measure $\delta_x$ defined by $\delta_x(x)=1$. The product of two probability distributions $P$ and $Q$ is denoted by $P\otimes Q$. A stochastic matrix with input alphabet $\X$ and output alphabet $\Z$ is written as a mapping $W:\X\rightarrow\P(\Z)$.  The $n$-fold memoryless extension of a channel $W:\X\rightarrow\P(\Z)$ is denoted by $W^{\otimes n}$, so that for $\x=(x_1,\ldots,x_n)\in\X^n$ and $\z=(z_1,\ldots,z_n)\in\Z^n$,
\[
  W^{\otimes n}(\z\vert\x)=\prod_{i=1}^nW(z_i\vert x_i).
\]
We also define for $P\in\P(\X)$ and $W:\X\rightarrow\P(\Z)$ the probability distribution $P\otimes W\in\P(\X\times\Z)$ by $(P\otimes W)(x,z)=P(x)W(z\vert x)$.

Every measure $\mu$ on the finite set $\X$ can be identified with a unique function $\mu:\X\rightarrow[0,\infty]$. Then for any subset $A\subset\X$ we have $\mu(A)=\sum_{x\in A}\mu(x)$. On the set of measures on $\X$, we define the total variation distance by
\[
  \lVert \mu_1-\mu_2\rVert:=\sum_{x\in\X}\lvert \mu_1(x)-\mu_2(x)\rvert.
\]

Given a random variable $X$ living on $\X$ and a $P\in\P(\X)$, we mean by $X\sim P$ that $P$ is the distribution of $X$. Given a pair of random variables $(X,Y)$ taking values in the finite set $\X\times\Y$, we write $P_X\in\P(\X)$ for the distribution of $X$ and $P_{X\vert Y}$ for the conditional distribution of $X$ given $Y$. We also write $T_{X,\delta}^n\subset\X^n$ for the subset of $\delta$-typical sequences with respect to $X$ and $T_{X\vert Y,\delta}^n(\y)\subset\Y^n$ for the subset of conditionally $\delta$-typical sequences with respect to $P_{X\vert Y}$ given $\y\in\Y^n$. Given a sequence $\x\in\X^n$ and an $x\in\X$, we let $N(x\vert\x)$ be the number of coordinates of $\x$ equal to $x$.

For random variables $X,Y,Z$ we write $H(X)$ for the entropy of $X$, $H(X\vert Y)$ for the conditional entropy of $X$ given $Y$, $I(X\wedge Y)$ for the mutual information of $X$ and $Y$ and $I(X\wedge Y\vert Z)$ for the conditional mutual information of $X$ and $Y$ given $Z$.

\textbf{Acknowledgment: }We would like to thank A. J. Pierrot for bringing the papers \cite{PB} and \cite{YRA} to our attention. They consider strong secrecy problems in multi-user settings with the help of resolvability theory. In particular, in \cite{YRA}, an achievable region for the wiretap MAC without common message or conferencing is derived.

\section{The Wiretap Multiple-Access Channel}\label{sect:wiretapMAC}

The wiretap Multiple-Access Channel (MAC) is described by a stochastic matrix
\[
  W:\X\times\Y\rightarrow\T\times\Z,
\]
where $\X,\Y,\T,\Z$ are finite sets. We write $W_b$ and $W_e$ for the marginal channels to $\T$ and $\Z$, so e.g.
\[
  W_b(t\vert x,y):=\sum_{z\in\Z}W(t,z\vert x,y).
\]
$\X$ and $\Y$ are the finite alphabets of Alice\eins\ and Alice\zwei, respectively. $\T$ is the finite alphabet of the receiver called Bob and the outputs received by the eavesdropper Eve are elements of the finite alphabet $\Z$. 

\subsection{With Common Message}\label{subsect:commmess}

Let $H_C$ be a nonnegative real number. A wiretap MAC code with common message and blocklength $n$ satisfying the  common randomness bound $H_C$ consists of a stochastic matrix
\begin{align*}
  G&:[K_0]\times[K_1]\times[K_2]\rightarrow\P(\X^n\times\Y^n)
\intertext{and a decoding function}
  \phi&:\T^n\rightarrow[K_0]\times[K_1]\times[K_2].
\end{align*}
$G$ is required to have the form
\[
  G(\x,\y\vert k_0,k_1,k_2)=\sum_{j\in\J}G_0(j\vert k_0)G_1(\x\vert k_0,k_1,j)G_2(\y\vert k_0,k_2,j),
\]
where $\J$ is some finite set and
\begin{align*}
  G_0&:[K_0]\rightarrow\P(\J),\\
  G_1&:[K_0]\times[K_1]\times\J\rightarrow\P(\X),\\
  G_2&:[K_0]\times[K_2]\times\J\rightarrow\P(\Y).
\end{align*}
Further, $G_0$ has to satisfy that $H(J\vert M_0)\leq nH_C$ for $M_0$ uniformly distributed on $[K_0]$ and $P_{J\vert M_0}=G_0$. $[K_0]$ is called the set of common messages, $[K_1]$ is the set of Alice\eins's private messages and $[K_2]$ the set of Alice\zwei's private messages. 

Let $M_0,M_1,M_2$ be independent random variables uniformly distributed on $[K_0]$, $[K_1]$ and $[K_2]$, respectively. Further, let $X^n,Y^n,T^n,Z^n$ be random variables such that for $(\x,\y,\t,\z)\in\X^n\times\Y^n\times\T^n\times\Z^n$
\begin{align*}
  P_{X^nY^n\vert M_0M_1M_2}(\x,\y\vert k_0,k_1,k_2)&=G(\x,\y\vert k_0,k_1,k_2),\\
  P_{T^nZ^n\vert X^nY^nM_0M_1M_2}(\t,\z\vert \x,\y,k_0,k_1,k_2)&=W^{\otimes n}(\t,\z\vert\x,\y).
\end{align*}
Then the average error of the code defined above equals
\[
  \PP[\phi(T^n)\neq (M_0,M_1,M_2)].
\]

\begin{definition}
A rate pair $(R_0,R_1,R_2)\in\RR_{\geq 0}^3$ is achievable by the wiretap MAC with common message under the common randomness bound $H_C\geq 0$ if for every $\eta>0$ and every $\eps\in(0,1)$ and $n$ large there exists a wiretap MAC code with common message and blocklength $n$ satisfying the common randomness bound $H_C$ and 
\begin{align*}
  \frac{1}{n}\log K_\nu&\geq R_\nu-\eta\qquad(\nu=0,1,2),\\
  \PP[\phi(T^n)\neq(M_0,M_1,M_2)]&\leq\eps,\\
  I(Z^n\wedge M_0M_1M_2)&\leq\eps.
\end{align*}
\end{definition}

\begin{remark}\label{ersterrem}
  It was shown in \cite{BBS} that no matter how Eve tries to decode the messages from the Alices, the average error must tend to one. More precisely, assume that a wiretap code with common message and blocklength $n$ is given, and assume that Eve has a decoding function 
\[
  \chi:\Z^n\rightarrow[K_0]\times[K_1]\times[K_2].
\]
Then
\[
  \PP[\chi(Z^n)\neq(M_0,M_1,M_2)]\geq 1-\eps'
\]
for some $\eps'$ which tends to zero as $\eps$ tends to zero. If $\eps$ tends to zero exponentially fast and $K_0,K_1,K_2$ grow exponentially, then $\eps'$ tends to zero at exponential speed.

More generally assume that $f:[K_0]\times[K_1]\times[K_2]\rightarrow[K']$ is a function satisfying $\PP[f(M_0,M_1,M_2)=k']=1/K'$ for all $k'\in[K']$. Then with the same argument as in \cite{BBS} one can show that for every function $g:\Z^n\rightarrow[K']$, one has $\PP[f(M)\neq g(Z^n)]\geq 1-1/K'-\eps'$ for the same $\eps'$ as above. That is, even for $K'=2$, blind guessing is the best way for Eve to estimate $f(M)$. In particular, no subset of the message random variables, like $M_0$ or $(M_1,M_2)$, can be reliably decoded by Eve.
\end{remark}

\subsection{With Conferencing Encoders}\label{subsect:confenc}

In the wiretap MAC with conferencing encoders, Alice\eins\ and Alice\zwei\ do not have a common message nor common randomness. However before forming their codewords, they may exchange some information about their private messages according to an iterative and randomized ``conferencing'' protocol whose deterministic form was introduced by Willems \cite{Wi1,Wi2}. If the respective message sets are $[K_1]$ and $[K_2]$, such a stochastic Willems conference can be described as follows. Let finite sets $\J_1$ and $\J_2$ be given which can be written as products
\begin{align*}
  \J_\nu=\J_{\nu,1}\times\ldots\times\J_{\nu,I}&&(\nu=1,2)
\end{align*}
for some positive integer $I$ which does not depend on $\nu$. A Willems conferencing stochastic matrix $c$ completely describing such a conference is determined in an iterative manner via sequences of stochastic matrices $c_{1,1},\ldots,c_{1,I}$ and $c_{2,1},\ldots,c_{2,I}$. $c_{1,i}$ describes the probability distribution of what Alice\eins\ tells Alice\zwei\ in the $i$-th conferencing iteration given the knowledge accumulated so far at Alice\eins. Thus in general, using the notation
\[
  \bar\nu:=\begin{cases}
             1\quad\text{if }\nu=2,\\
             2\quad\text{if }\nu=1,
           \end{cases}
\]
these stochastic matrices satisfy for $\nu=1,2$ and $i=2,\ldots,I$,
\begin{align*}
  c_{\nu,1}&:[K_\nu]\rightarrow\P(\J_{\nu,1}),\\
  c_{\nu,i}&:[K_\nu]\times\J_{\bar\nu,1}\times\ldots\times\J_{\bar\nu,i-1}\rightarrow\P(\J_{\nu,i}).
\end{align*}
The stochastic matrix $c:[K_1]\times[K_2]\rightarrow\P(\J_1\times\J_2)$ is obtained by setting 
\begin{align*}
  &\mathrel{\hphantom{=}}\;c(j_{1,1},\ldots,j_{1,I},j_{2,1},\ldots,j_{2,I}\vert k_1,k_2)\\
  &:=c_{1,1}(j_{1,1}\vert k_1)\,c_{2,1}(j_{2,1}\vert k_2)\cdots\\
  &\mathrel{\hphantom{=}}\;\cdots c_{1,I}(j_{1,I}\vert k_1,j_{2,1},\ldots,j_{2,I-1})\,c_{2,I}(j_{2,I}\vert k_2,j_{1,1},\ldots,j_{1,I-1}).
\end{align*}
We denote the $\J_1$- and $\J_2$-marginals of this stochastic matrix by $c_1$ and $c_2$, so $c_1(j_{1,1},\ldots,j_{1,I}\vert k_1,k_2)$ is obtained by summing over $j_{2,1},\ldots,j_{2,I}$ and $c_2$ is obtained analogously.

Now we define a wiretap MAC code with conferencing encoders. It consists of a Willems conferencing stochastic matrix $c:[K_1]\times[K_2]\rightarrow\P(\J_1\times\J_2)$ as above together with encoding stochastic matrices
\begin{align*}
  G_1&:[K_1]\times\J_2\rightarrow\X^n,\\
  G_2&:[K_2]\times\J_1\rightarrow\Y^n
\intertext{and a decoding function}
  \phi&:\T^n\rightarrow[K_1]\times[K_2].
\end{align*}
$[K_1]$ is the set of Alice\eins's messages and $[K_2]$ is the set of Alice\zwei's messages. A pair $(k_1,k_2)\in[K_1]\times[K_2]$ is encoded into the codeword pair $(\x,\y)\in\X^n\times\Y^n$ with probability
\begin{equation}\label{XY-distr}
  \sum_{(j_1,j_2)\in\J_1\times\J_2}c(j_1,j_2\vert k_1,k_2)\,G_1(\x\vert k_1,j_2)\,G_2(\y\vert k_2,j_1).
\end{equation}
In particular, conferencing generates common randomness. As both $c_1$ and $c_2$ may depend on both encoders' messages, the codewords may as well depend on both messages. Thus if conferencing were unrestricted, this would transform the MAC into a single-user wiretap channel with input alphabet $\X\times\Y$. However, Willems introduces a restriction in terms of the blocklength of the code which is used for transmission. For conferencing under conferencing capacities $C_1,C_2\geq 0$, he requires that for a blocklength-$n$ code, $\lvert\J_1\rvert$ and $\lvert\J_2\rvert$ satisfy
\begin{align}\label{eq:Willemsrates}
  \frac{1}{n}\log\lvert\J_\nu\rvert\leq C_\nu&&(\nu=1,2).
\end{align}
We also impose this constraint and define a wiretap MAC code with conferencing capacities $C_1,C_2\geq 0$ to be a wiretap MAC code with conferencing encoders satisfying \eqref{eq:Willemsrates}.

Let a wiretap MAC code with conferencing encoders be given and let $M_1,M_2$ be independent random variables uniformly distributed on $[K_1]$ and $[K_2]$, respectively. Let $X^n,Y^n,T^n,Z^n$ be random variables such that conditional on $(M_1,M_2)$, the distribution of $(X^n,Y^n)$ is given by \eqref{XY-distr} and such that 
\[
  P_{T^nZ^n\vert X^nY^nM_1M_2}=W^{\otimes n}.
\]
Then the average error of the code defined above equals
\[
  \PP[\phi(T^n)\neq (M_1,M_2)].
\]

\begin{definition}
A rate pair $(R_1,R_2)\in\RR_{\geq0}^3$ is achievable by the wiretap MAC with conferencing encoders at conferencing capacities $C_1,C_2>0$ if for every $\eta>0$ and every $\eps\in(0,1)$ and for $n$ large there exists a wiretap MAC code with conferencing capacities $C_1,C_2$ and blocklength $n$ satisfying
\begin{align*}
  \frac{1}{n}\log K_\nu&\geq R_\nu-\eta\qquad(\nu=1,2),\\
  \PP[\phi(T^n)\neq(M_1,M_2)]&\leq\eps,\\
  I(Z^n\wedge M_1M_2)&\leq\eps.
\end{align*}
\end{definition}

\begin{remark}
  Here again, as in Remark \ref{ersterrem}, the average decoding error for any decoder Eve might apply tends to 1 if the security criterion is satisfied.
\end{remark}

\section{Coding Theorems}\label{sect:thethms}

\subsection{For the Wiretap MAC with Common Message}\label{subsect:commmessthm}

Let $H_C\geq 0$ be the common randomness bound. The rate region whose achievability we are about to claim in Theorem \ref{thmcomm} can be written as the closure of the convex hull of the union of certain rate sets which are parametrized by the elements of a subset $\Pi_{H_C}$ of the set $\Pi$ which is defined as follows. $\Pi$ contains all probability distributions $p$ of random vectors $(U,V_1,V_2,X,Y,T,Z)$ living on sets $\U\times\V_1\times\V_2\times\X\times\Y\times\T\times\Z$, where $\U,\V_1,\V_2$ are finite subsets of the integers and where $p$ has the form
\[
  p=P_U\otimes P_{V_1\vert U}\otimes P_{V_2\vert U}\otimes P_{X\vert V_1}\otimes P_{Y\vert V_2}\otimes W.
\]

Next we define $\Pi_{H_C}$. There are four cases altogether, numbered Case 0 to Case 3. Case 0 corresponds to $H_C=0$ and if $H_C>0$, then $\Pi_{H_C}$ has the form $\Pi_{H_C}=\Pi_{H_C}\1\cup\Pi_{H_C}\2\cup\Pi_{H_C}\3$, and each of these subsets corresponds to one of these cases. The one condition all cases have in common is that $I(Z\wedge V_1V_2)\leq I(T\wedge V_1V_2)$.

\textbf{Case 0:} If $H_C=0$ define the set $\Pi_0$ as the set of those $p\in\Pi$ where $V_1$ and $V_2$ are independent of $U$ (so we can omit $U$ in this case and $V_1$ and $V_2$ are independent) and where $p$ satisfies the inequalities
\begin{align}
  I(Z\wedge V_1)&\leq I(T\wedge V_1\vert V_2),\label{HC01}\\
  I(Z\wedge V_2)&\leq I(T\wedge V_2\vert V_1).\label{HC02}
\end{align}
For $p\in\Pi_0$ define the set $\R\0(p)$ to be the set of nonnegative triples $(R_0,R_1,R_2)$ satisfying
\begin{align*}
  R_0&=0,\\
  R_1&\leq I(T\wedge V_1\vert V_2)-I(Z\wedge V_1)-[I(Z\wedge V_2\vert V_1)-I(T\wedge V_2\vert V_1)]_+,\\
  R_2&\leq I(T\wedge V_2\vert V_1)-I(Z\wedge V_2)-[I(Z\wedge V_1\vert V_2)-I(T\wedge V_1\vert V_2)]_+,\\
  R_1+R_2&\leq I(T\wedge V_1V_2)-I(Z\wedge V_1V_2).
\end{align*}

\textbf{Case 1:} $\Pi_{H_C}\1$ is the set of those $p\in\Pi$ which satisfy $I(Z\wedge U)<H_C$ and
\begin{align}
  I(Z\wedge V_1\vert U)&\leq I(T\wedge V_1\vert V_2U),\label{compi1}\\
  I(Z\wedge V_2\vert U)&\leq I(T\wedge V_2\vert V_1U),\label{compi2}\\
  I(Z\wedge V_1V_2\vert U)&\leq I(T\wedge V_1\vert V_2U)+I(T\wedge V_2\vert V_1U).\label{compi3}
\end{align}
Then we denote by $\R\1(p)$ the set of nonnegative real triples $(R_0,R_1,R_2)$ satisfying
\begin{align*}
  R_1&\leq I(T\wedge V_1\vert V_2U)-I(Z\wedge V_1\vert U)\\
      &\hspace{3.5cm}-[I(Z\wedge V_2\vert V_1U)-I(T\wedge V_2\vert V_1U)]_+,\\
  R_2&\leq I(T\wedge V_2\vert V_1U)-I(Z\wedge V_2\vert U)\\
      &\hspace{3.5cm}-[I(Z\wedge V_1\vert V_2U)-I(T\wedge V_1\vert V_2U)]_+,\\
  R_1+R_2&\leq I(T\wedge V_1V_2\vert U)-I(Z\wedge V_1V_2\vert U),\\
  R_0+R_1+R_2&\leq I(T\wedge V_1V_2)-I(Z\wedge V_1V_2).
\end{align*}

\textbf{Case 2:} The conditions for $p$ to be contained in $\Pi_{H_C}\2$ cannot be phrased as simply as for $\Pi_{H_C}\1$. Generally, if $p\in\Pi_{H_C}\2$ then 
\[
  \min\{I(Z\wedge V_1U),I(Z\wedge V_2U)\}<H_C\leq I(Z\wedge V_1V_2).
\]
This is sufficient if  $I(Z\wedge V_1\vert V_2U)=I(Z\wedge V_2\vert V_1U)$. If $I(Z\wedge V_1\vert V_2U)> I(Z\wedge V_2\vert V_1U)$ then we additionally require that
\begin{align*}
  &\mathrel{\hphantom{\leq}}\alpha\2_0:=\max\left(\frac{I(Z\wedge V_1U)-H_C}{I(Z\wedge V_1\vert V_2U)-I(Z\wedge V_2\vert V_1U)},1-\frac{I(T\wedge V_2\vert V_1U)}{I(Z\wedge V_2\vert V_1U)},0\right)\\
  &\leq\alpha\2_1:=\min\left(\frac{I(T\wedge V_1\vert V_2U)}{I(Z\wedge V_1\vert V_2U)},\frac{I(T\wedge V_1V_2\vert U)-I(Z\wedge V_2\vert V_1U)}{I(Z\wedge V_1\vert V_2U)-I(Z\wedge V_2\vert V_1U)},1\right)
\end{align*}
whereas if $I(Z\wedge V_1\vert V_2U)<I(Z\wedge V_2\vert V_1U)$ then we need
\begin{align*}
  &\mathrel{\hphantom{\leq}}\alpha\2_0:=\max\left(1-\frac{I(T\wedge V_2\vert V_1U)}{I(Z\wedge V_2\vert V_1U)},\frac{I(T\wedge V_1V_2\vert U)-I(Z\wedge V_2\vert V_1U)}{I(Z\wedge V_1\vert V_2U)-I(Z\wedge V_2\vert V_1U)},0\right)\\
  &\leq\alpha\2_1:=\min\left(\frac{H_C-I(Z\wedge V_1U)}{I(Z\wedge V_2\vert V_1U)-I(Z\wedge V_1\vert V_2U)},\frac{I(T\wedge V_1\vert V_2U)}{I(Z\wedge V_1\vert V_2U)},1\right).
\end{align*}
In the case of equality, i.e. if $I(Z\wedge V_1\vert V_2U)=I(Z\wedge V_2\vert V_1U)$, we define $\R\2(p)$ as
\begin{align*}
  R_1&\leq I(T\wedge V_1\vert V_2U),\\
  R_2&\leq I(T\wedge V_2\vert V_1U),\\
  R_1+R_2&\leq I(T\wedge V_1V_2\vert U)-I(Z\wedge V_1\vert V_2U),\\
  R_0+R_1+R_2&\leq I(T\wedge V_1V_2)-I(Z\wedge V_1V_2).
\end{align*}
If $I(Z\wedge V_1\vert V_2U)>I(Z\wedge V_2\vert V_1U)$, we define $\R\2(p)$ by
\begin{align}
  R_1&\leq I(T\wedge V_1\vert V_2U)-\alpha\2_0I(Z\wedge V_1\vert V_2U),\notag\\
  R_2&\leq I(T\wedge V_2\vert V_1U)-(1-\alpha\2_1)I(Z\wedge V_2\vert V_1U),\notag\\
  R_1+R_2&\leq I(T\wedge V_1V_2\vert U)-\alpha\2_0I(Z\wedge V_1\vert V_2U)\label{sumbound}\\&\hspace{4cm}-(1-\alpha\2_0)I(Z\wedge V_2\vert V_1U),\notag\\
  R_1+\frac{I(Z\wedge V_2\vert V_1U)}{I(Z\wedge V_1\vert V_2U)}R_2&\leq I(T\wedge V_2\vert V_1U)\label{wsumbound}\\&\hspace{1.8cm}+\left(\frac{I(T\wedge V_1\vert U)}{I(Z\wedge V_1\vert V_2U)}-1\right)I(Z\wedge V_2\vert V_1U),\notag\\
  R_0+R_1+R_2&\leq I(T\wedge V_1V_2)-I(Z\wedge V_1V_2).\notag
\end{align}
The bound \eqref{sumbound} on $R_1+R_2$ can be reformulated as
\begin{align*}
  R_1+R_2&\leq I(T\wedge V_1V_2\vert U)-I(Z\wedge V_1V_2\vert U)\\
  &\mathrel{\hphantom{\leq}}+\min\biggl\{H_C-I(Z\wedge U),I(Z\wedge V_1\vert U),\\
  &\mathrel{\hphantom{\leq}}\hspace{1.2cm}I(T\wedge V_1\vert V_2U)\left(\frac{I(Z\wedge V_2\vert V_1U)}{I(Z\wedge V_1\vert V_2U)}-1\right)+I(Z\wedge V_1\vert U)\biggr\},
\end{align*}
and if $I(Z\wedge V_2\vert V_1U)>0$, we can give the weighted sum bound \eqref{wsumbound} the almost symmetric form
\begin{align*}
  \frac{R_1}{I(Z\wedge V_1\vert V_2U)}+\frac{R_2}{I(Z\wedge V_2\vert V_1U)}&\leq\frac{I(T\wedge V_1\vert U)}{I(Z\wedge V_1\vert V_2U)}+\frac{I(T\wedge V_2\vert V_1U)}{I(Z\wedge V_2\vert V_1U)}-1.
\end{align*}
For the case that $I(Z\wedge V_1\vert V_2U)<I(Z\wedge V_2\vert V_1U)$, we define $\R\2(p)$ by exchanging the roles of $V_1$ and $V_2$.

\textbf{Case 3:} We define $\Pi_{H_C}\3$ to be the set of those $p\in\Pi$ with $I(Z\wedge V_1V_2)<H_C$ and for such a $p$ let $\R\3(p)$ equal
\begin{align*}
  R_1&\leq I(T\wedge V_1\vert V_2U),\\
  R_2&\leq I(T\wedge V_2\vert V_1U),\\
  R_1+R_2&\leq I(T\wedge V_1V_2\vert U),\\
  R_0+R_1+R_2&\leq I(T\wedge V_1V_2)-I(Z\wedge V_1V_2).
\end{align*}

\begin{theorem}\label{thmcomm}
  For the common randomness bound $H_C=0$, the wiretap MAC $W$ with common message achieves the set 
\begin{equation}\label{zero-achi}
  closure\Biggl(conv\Biggl(\bigcup_{p\in\Pi_{0}}\R\0(p)\Biggr)\Biggr).
\end{equation}
If $H_C>0$, then the closure of the convex hull of the set 
\[
  \bigcup_{p\in\Pi\1_{H_C}}\!\!\R\1(p)\;\cup\bigcup_{p\in\Pi\2_{H_C}}\!\!\R\2(p)\;\cup\bigcup_{p\in\Pi\3_{H_C}}\!\!\R\3(p)
\]
is achievable.
\end{theorem}

\begin{remark}\label{rem:auxcard}
  Using the standard Carath\'eodory-Fenchel technique, one can show that one may without loss of generality assume $\lvert\U\rvert\leq\lvert\X\rvert\lvert\Y\rvert+5$. However, $\lvert\V_1\rvert$ and $\lvert\V_2\rvert$ cannot be bounded in this way, as the application of the Carath\'eodory-Fenchel theorem does not preserve the conditional independence of $V_1$ and $V_2$. Thus a characterization of the above achievable region involving sets with upper-bounded cardinality is currently not available. As it would be important for an efficient calculation of the achievable region, it still requires further consideration. 
\end{remark}

\begin{remark}\label{rem:nocommrandnocommmess}
  If no common randomness is available, then no common message can be transmitted.
\end{remark}

\begin{remark}\label{rem:suffcommrand}
  We have $\R\1(p)\subset\R\2(p)\subset\R\3(p)$. This can be seen directly at the beginning of the proof in Subsection \ref{subsect:elrateregs} where we decompose the regions $\R^{(\nu)}(p)$ for $\nu=1,2$ into a union of simpler regions. 

In particular, if $H_C$ is larger than the capacity of the single-sender discrete memoryless channel $W_e$ with input alphabet $\X\times\Y$ and output alphabet $\Z$, then $\Pi_{H_C}\3=\Pi$ and the achievable set equals
\[
  closure\Biggl(conv\Biggl(\bigcup_{p\in\Pi}\R\3(p)\Biggr)\Biggr).
\]
In this case the maximal sum rate equals
\begin{equation}\label{eq:single-sender-secrecy-cap}
  \C:=\max_{p\in\Pi}\bigl(I(T\wedge V_1V_2)-I(Z\wedge V_1V_2)\bigr).
\end{equation}
This equals the secrecy capacity of the single-sender wiretap channel when Alice\eins\ and Alice\zwei\ together are considered as one single sender. In order to see this, we have to show that for any pair $(V_1',V_2')$ of random variables on any Cartesian product $\V_1\times\V_2$ of finite sets one can find random variables $(V_1,V_2,U)$ satisfying $P_{UV_1V_2}=P_U\otimes (P_{V_1\vert U}\otimes P_{V_2\vert U})$ and $P_{V_1V_2}=P_{V_1'V_2'}$. Given such arbitrary $(V_1',V_2')$ as above, just define $U=V_1'$ and $P_{V_1\vert U}=$ the identity on $\V_1$ and $P_{V_2\vert U}=P_{V_2'\vert V_1'}$. Then a simple calculation shows that the above conditions are satisfied. Thus \eqref{eq:single-sender-secrecy-cap} equals the secrecy capacity of the single-sender wiretap channel with Alice\eins\ and Alice\zwei\ combined into a single sender. The remaining conditions on $R_1$ and $R_2$ formulated in the definition of $\R\3(p)$ are not concerned with $W_e$, they are required by the non-wiretap MAC coding theorem applied 
to $W_b$.
\end{remark}

\subsection{For the Wiretap MAC with Conferencing Encoders}\label{subsect:confencthm}

For conferencing capacities $C_1,C_2>0$, the achievable rate region is parametrized by the members of $\Pi_{C_1+C_2}$. We have Cases 1-3 from the common message part. 

\textbf{Case 1:} For $p\in\Pi_{C_1+C_2}\1$ we define $\R\1(p,C_1,C_2)$ by
\begin{align*}
  R_1&\leq I(T\wedge V_1\vert V_2U)-I(Z\wedge V_1\vert U)\\
  &\mathrel{\hphantom{\leq}}-[I(Z\wedge V_2\vert V_1U)-I(T\wedge V_2\vert V_1U)]_++C_1-[I(Z\wedge U)-C_2]_+,\\
  R_2&\leq I(T\wedge V_2\vert V_1U)-I(Z\wedge V_2\vert U)\\
  &\mathrel{\hphantom{\leq}}-[I(Z\wedge V_1\vert V_2U)-I(T\wedge V_1\vert V_2U)]_++C_2-[I(Z\wedge U)-C_1]_+,\\
  R_1+R_2&\leq\min\{I(T\wedge V_1V_2\vert U)+C_1+C_2, I(T\wedge V_1V_2))\}-I(Z\wedge V_1V_2)).
\end{align*}

\textbf{Case 2:} For $p\in\Pi_{C_1+C_2}\2$, we set $J_0\halpha:=\alpha I(Z\wedge V_2U)+(1-\alpha)I(Z\wedge V_1U)$. For $\alpha\in[\alpha\2_0,\alpha\2_1]$ define the set $\R\2_{\alpha}(p,C_1,C_2)$ by
\begin{align*}
  R_1&\leq I(T\wedge V_1\vert V_2U)-\alpha I(Z\wedge V_1\vert V_2U)+C_1-[J_0\halpha-C_2]_+,\\
  R_2&\leq I(T\wedge V_2\vert V_1U)-(1-\alpha)I(Z\wedge V_2\vert V_1U)+C_2-[J_0\halpha-C_1]_+,\\
  R_1+R_2&\leq\min\{I(T\wedge V_1V_2\vert U)+C_1+C_2,I(T\wedge V_1V_2)\}-I(Z\wedge V_1V_2).
\end{align*}
Then we set 
\[
  \R\2(p,C_1,C_2):=\bigcup_{\alpha\2_0\leq\alpha\leq\alpha\2_1}\R\2_{\alpha}(p,C_1,C_2).
\]

\textbf{Case 3:} For $p\in\Pi_{C_1+C_2}\3$ we define $\R\3(p,C_1,C_2)$ by
\begin{align*}
  R_1&\leq I(T\wedge V_1\vert V_2U)+C_1-[I(Z\wedge V_1V_2)-C_2]_+,\\
  R_2&\leq I(T\wedge V_2\vert V_1U)+C_2-[I(Z\wedge V_1V_2)-C_1]_+,\\
  R_1+R_2&\leq \min\{I(T\wedge V_1V_2\vert U)+C_1+C_2,I(T\wedge V_1V_2)\}-I(Z\wedge V_1V_2).
\end{align*}

\begin{theorem}\label{thmconf}
  For the conferencing capacities $C_1,C_2>0$, the wiretap MAC $W$ with conferencing encoders achieves the closure of the convex hull of the set 
\[
  \bigcup_{p\in\Pi\1_{H_C}}\!\!\R\1(p,C_1,C_2)\;\cup\bigcup_{p\in\Pi\2_{H_C}}\!\!\R\2(p,C_1,C_2)\;\cup\bigcup_{p\in\Pi\3_{H_C}}\!\!\R\3(p,C_1,C_2).
\]
\end{theorem}

\begin{remark}
  Remark \ref{rem:auxcard} applies here, too.
\end{remark}

\begin{remark}
  The stochastic conferencing protocols employed to achieve the sets in Theorem \ref{thmconf} are non-iterative. That means that the $c$ we use in the proof have the form
\[
  c(v_1,v_2\vert k_1,k_2)=c_1(v_1\vert k_1)c_2(v_2\vert k_2).
\]
\end{remark}

\begin{remark}\label{remzeroachi}
  If $C_1=C_2=0$, then the maximal rate set whose achievability we can show is \eqref{zero-achi}. Conferencing only enlarges this set in the presence of a wiretapper if it is used to establish common randomness between the encoders. At least this is true for the achievable region we can show, it cannot be verified in general as long as one does not have a converse. The reason is that conferencing generates a common message shared by Alice\eins\ and Alice\zwei. As noted in Remark \ref{rem:nocommrandnocommmess}, a common message can only be kept secret if common randomness is available. As the Alices do not have common randomness a priori, this also has to be generated by conferencing, so the Willems conferencing protocol has to be stochastic. 
\end{remark}

\begin{remark}
With the coding method we apply, conferencing may enable secure transmission if this is not possible without. That means that there are wiretap MACs where the achievable region without conferencing as derived in Theorem \ref{thmcomm} only contains the rate pair $(0,0)$ whereas it contains non-trivial rate pairs with $C_1,C_2>0$. See Section \ref{sect:discussion} for an example.
\end{remark}

\begin{remark}
If $C_1,C_2$ are sufficiently large, then the maximal achievable sum rate equals the secrecy capacity $\C$ of the single-sender wiretap channel with input alphabet $\X\times\Y$ and channel matrix $W$, see \eqref{eq:single-sender-secrecy-cap}. In fact, this happens if 
\begin{enumerate}[1)]
  \item $C_1+C_2$ is strictly larger than the capacity of the single-sender discrete memoryless channel $W_e$ with input alphabet $\X\times\Y$ and output alphabet $\Z$,
  \item $C_1+C_2\geq\min_{p\in\Pi^*}I(T\wedge U)$, where $\Pi^*$ contains those $p\in\Pi$ which achieve $\C$.
\end{enumerate}
Condition 1) is sufficient to guarantee that $\C$ is achievable by an element of $\Pi_{C_1+C_2}\3$ which then equals $\Pi$, see Remark \ref{rem:suffcommrand}. In particular $\Pi^*$ is nonempty, and 2) ensures that the maximum over $\Pi$ of the sum rate bounds from $\R\3(p,C_1,C_2)$ equals $\C$. 
\end{remark}

\section{Proof of Theorem \ref{thmcomm}}\label{proofcomm}

\subsection{Elementary Rate Regions}\label{subsect:elrateregs}

For Cases 0, 1 and 2 we first show the achievability of certain rate regions whose union or convex combination then yields the achievable regions claimed in the theorem. 

\subsubsection{For Case 0 and 1:}

We only consider Case 1, Case 0 is analogous. The considerations hold for $I(Z\wedge V_1\vert U)<I(Z\wedge V_1\vert V_2U)$ which is equivalent to $I(Z\wedge V_2\vert U)<I(Z\wedge V_2\vert V_1U)$. In the case of equality we can prove the achievability of $\R(p)$ directly. Define 
\begin{align*}
  \alpha\1_0&:=\biggl[\frac{I(T\wedge V_2\vert V_1U)-I(Z\wedge V_2\vert V_1U)}{I(Z\wedge V_2\vert U)-I(Z\wedge V_2\vert V_1U)}\biggr]_+,\\
  \alpha\1_1&:=\min\biggl\{\frac{I(T\wedge V_1\vert V_2U)-I(Z\wedge V_1\vert U)}{I(Z\wedge V_1\vert V_2U)-I(Z\wedge V_1\vert U)},1\biggr\}.
\end{align*}
Note that conditions \eqref{compi1}-\eqref{compi3} are equivalent to $\alpha\1_0\leq\alpha\1_1$. For $\alpha\in[\alpha\1_0,\alpha\1_1]$ we define a rate region $\R\1_\alpha(p)$ by the bounds
\begin{align*}
  R_1&\leq I(T\wedge V_1\vert V_2U)-\alpha I(Z\wedge V_1\vert V_2U)-(1-\alpha)I(Z\wedge V_1\vert U),\\
  R_2&\leq I(T\wedge V_2\vert V_1U)-\alpha I(Z\wedge V_2\vert U)-(1-\alpha)I(Z\wedge V_2\vert V_1U),\\
  R_1+R_2&\leq I(T\wedge V_1V_2\vert U)-I(Z\wedge V_1V_2\vert U),\\
  R_0+R_1+R_2&\leq I(T\wedge V_1V_2)-I(Z\wedge V_1V_2).
\end{align*}

\begin{lemma}\label{alpha-union}
We have
\begin{equation*}
  \R\1(p)=\bigcup_{\alpha\1_0\leq\alpha\leq\alpha\1_1}\R\1_\alpha(p).
\end{equation*}
Thus if $\R\1_\alpha(p)$ is an achievable rate region for every $\alpha\in[\alpha\1_0,\alpha\1_1]$, then $\R\1(p)$ is achievable.
\end{lemma}

For the proof we use the following lemma which is proved in the appendix.

\begin{lemma}\label{gemconc}
  Assume that $a_1,a_2,b_1,b_2,c,d,r_1,r_2,r_{12},r_{012}$ are nonnegative reals satisfying
\begin{align*}
  a_1> b_1,\quad a_2< b_2,\quad a_1+a_2=b_1+b_2=c,\quad r_1+r_2\geq r_{12}.
\end{align*}
Let $0\leq\alpha_0\leq\alpha_1\leq1$. For every $\alpha\in[\alpha_0,\alpha_1]$, let a three-dimensional convex subset $\K_\alpha$ of $\mathbb{R}_{\geq 0}^3$ be defined by
\begin{align*}
  R_1&\leq r_1-\alpha a_1-(1-\alpha)b_1,\\
  R_2&\leq r_2-\alpha a_2-(1-\alpha)b_2,\\
  R_1+R_2&\leq r_{12}-c,\\
  R_0+R_1+R_2&\leq r_{012}-d
\end{align*}
and assume that $\K_\alpha\neq\varnothing$ for every $\alpha$. Then
\begin{equation}\label{statementalpha-union}
  \bigcup_{\alpha_0\leq\alpha\leq\alpha_1}\K_\alpha=\K,
\end{equation}
where $\K$ is defined by
\begin{align}
  R_1&\leq r_1-\alpha_0a_1-(1-\alpha_0)b_1,\label{set1}\\
  R_2&\leq r_2-\alpha_1a_2-(1-\alpha_1)b_2,\label{set2}\\
  R_1+R_2&\leq r_{12}-c,\label{set4}\\
  R_0+R_1+R_2&\leq r_{012}-d.\label{set3} 
\end{align}
\end{lemma}

\begin{proof}[Lemma \ref{alpha-union}]
  The proof is a direct application of Lemma \ref{gemconc} by setting
\begin{align*}
  r_1&=I(T\wedge V_1\vert V_2U), & r_2&=I(T\wedge V_2\vert V_1U),\\
  r_{12}&=I(T\wedge V_1V_2\vert U), & r_{012}&=I(T\wedge V_1V_2),\\
  a_1&=I(Z\wedge V_1\vert V_2U), & a_2&=I(Z\wedge V_2\vert U),\\
  b_1&=I(Z\wedge V_1\vert U), & b_2&=I(Z\wedge V_2\vert V_1U),\\
  \alpha_0&=\alpha\1_0, & \alpha_1&=\alpha\1_1.
\end{align*}
We just need to show that the bounds \eqref{set1} and \eqref{set2} coincide with those from the definition of $\R\1(p)$. This is easy for the case $\alpha\1_0=0$ because in that case we have $I(T\wedge V_2\vert V_1U)\geq I(Z\wedge V_2\vert V_1U)$ and the positive part in the bound on $R_1$ in the definition of $\R\1(p)$ vanishes. Similarly $\alpha\1_1=1$ implies $I(T\wedge V_1\vert V_2U)\geq I(Z\wedge V_1\vert V_2U)$ and the positive part in the bound on $R_2$ in the definition of $\R\1(p)$ vanishes. Now assume that $\alpha\1_0>0$. This assumption implies $I(Z\wedge V_2\vert V_1U)>I(T\wedge V_2\vert V_1U)$. Thus we obtain for the equivalent of \eqref{set1}
\begin{align*}
  &\;\quad I(T\wedge V_1\vert V_2U)-I(Z\wedge V_1\vert U)\\
  &\qquad-\frac{I(T\wedge V_2\vert V_1U)-I(Z\wedge V_2\vert V_1U)}{I(Z\wedge V_2\vert U)-I(Z\wedge V_2\vert V_1U)}
  (I(Z\wedge V_1\vert V_2U)-I(Z\wedge V_1\vert U))\\
  &=I(T\wedge V_1\vert V_2U)-I(Z\wedge V_1\vert U)\\
  &\qquad-\frac{I(T\wedge V_2\vert V_1U)-I(Z\wedge V_2\vert V_1U)}{I(Z\wedge V_2\vert U)-I(Z\wedge V_2\vert V_1U)}
  (I(Z\wedge V_2\vert V_1U)-I(Z\wedge V_2\vert U))\\
  &=I(T\wedge V_1\vert V_2U)+I(T\wedge V_2\vert V_1U)-I(Z\wedge V_1V_2\vert U)\\
  &=I(T\wedge V_1\vert V_2U)-I(Z\wedge V_1\vert U)-[I(Z\wedge V_2\vert V_1U)-I(T\wedge V_2\vert V_1U)]_+.
\end{align*}
If $\alpha\1_1<1$, we obtain the analog for the bound on $R_2$. This shows with Lemma \ref{gemconc} that $\R\1(p)$ can be represented as the union of the sets $\R\1_\alpha(p)$ for $\alpha\1_0\leq\alpha\leq\alpha\1_1$.\qed
\end{proof}

\subsubsection{For Case 2:} 

Here we assume that $I(Z\wedge V_1\vert V_2U)\neq I(Z\wedge V_2\vert V_1U)$ which is equivalent to $I(Z\wedge V_1U)\neq I(Z\wedge V_2U)$. In the case of equality, the achievability of $\R\2(p)$ can be shown directly. Define for $\alpha\in[\alpha\2_0,\alpha\2_1]$ the rate set $\R\2_\alpha(p)$ by the conditions
\begin{align*}
  R_1&\leq I(T\wedge V_1\vert V_2U)-\alpha I(Z\wedge V_1\vert V_2U),\\
  R_2&\leq I(T\wedge V_2\vert V_1U)-(1-\alpha)I(Z\wedge V_2\vert V_1U),\\
  R_1+R_2&\leq I(T\wedge V_1V_2\vert U)-\alpha I(Z\wedge V_1\vert V_2U)-(1-\alpha)I(Z\wedge V_2\vert V_1U),\\
  R_0+R_1+R_2&\leq I(T\wedge V_1V_2)-I(Z\wedge V_1V_2).
\end{align*}

\begin{lemma}\label{lemunioncase2}
  We have that
\[
  \R\2(p)=\bigcup_{\alpha\2_0\leq\alpha\leq\alpha\2_1}\R\2_\alpha(p).
\]
In particular, if $\R\2_\alpha(p)$ is achievable for every $\alpha\in[\alpha\2_0,\alpha\2_1]$, then so is $\R\2(p)$.
\end{lemma}

\begin{remark}
  The similarity between the rate regions for Case 1 and Case 2 becomes clear in these decompositions. The description for Case 2 is more complex because $\alpha\2_0$ and $\alpha\2_1$ are defined through three minima/maxima. This is due to the fact that the sum $\alpha I(Z\wedge V_1\vert V_2U)+(1-\alpha)I(Z\wedge V_2\vert V_1U)$ is not constant in $\alpha$. Hence the conditions for $\alpha\2_0\leq\alpha\2_1$ cannot be reformulated into simple conditions on the corresponding $p$.
\end{remark}

One obtains Lemma \ref{lemunioncase2} from the next lemma by making the following replacements:
\begin{align*}
  r_1&=I(T\wedge V_1\vert V_2U),& r_2&=I(T\wedge V_2\vert V_1U),\\
  r_{12}&=I(T\wedge V_1V_2\vert U),& r_{012}&=I(T\wedge V_1V_2),\\
  a&=I(Z\wedge V_1\vert V_2U),&b&=I(Z\wedge V_2\vert V_1U),\\
  c&=I(Z\wedge V_1V_2),\\
  \alpha_0&=\alpha\2_0, & \alpha_1&=\alpha\2_1.
\end{align*}

\begin{lemma}\label{unionconv2}
  Let $r_1,r_2,r_{12},r_{012},a,b,c$ be nonnegative reals with $\max(r_1,r_2)\leq r_{12}\leq r_1+r_2$. Let $\alpha_0,\alpha_1\in[0,1]$ be given such that for every $\alpha\in[\alpha_0,\alpha_1]$ the set $\K_\alpha$ defined by 
\begin{align*}
  R_1&\leq r_1-\alpha a,\\
  R_2&\leq r_2-(1-\alpha)b,\\
  R_1+R_2&\leq r_{12}-\alpha a-(1-\alpha)b,\\
  R_0+R_1+R_2&\leq r_{012}-c
\end{align*}
is nonempty. If $a\leq b$, the convex hull of the union of these sets is given by the set $\K$ which is characterized by
\begin{align}
  0\leq R_1&\leq r_1-\alpha_0a,\label{conv21}\\
  0\leq R_2&\leq r_2-(1-\alpha_1)b,\label{conv22}\\
  R_1+R_2&\leq r_{12}-\alpha_1a-(1-\alpha_1)b,\label{conv23}\\
  bR_1+aR_2&\leq r_{12}a+r_1(b-a)-ab,\label{conv24}\\
  R_0+R_1+R_2&\leq r_{012}-c.\label{conv25}
\end{align}
If $a>b$, the convex hull of the union of the sets $\K_\alpha$ is given by analogous bounds where $a$ and $b$ are exchanged in \eqref{conv24}.
\end{lemma}

The proof of Lemma \ref{unionconv2} can be found in the appendix.

\subsection{How to Prove Secrecy}\label{subsect:provesec}

Proving secrecy using Chernoff-type concentration inequalities (see Subsection \ref{sect:setup}) is the core of Devetak's approach to the wiretap channel \cite{De05}. Due to the multi-user structure of the inputs of the wiretap MAC, we need several such Chernoff-type inequalities basing on each other compared to the one used by Devetak (actually an application of the Ahlswede-Winter lemma). However, once these are established, the way of obtaining secrecy is exactly the same as presented by Devetak. With the help of the inequalities one obtains a code with stochastic encoding and a measure $\vartheta$ (not necessarily a probability measure!) such that for all $k_0,k_1,k_2$
\begin{equation}\label{totvarsmall}
  \lVert P_{Z^{n}\vert M_0=k_0,M_1=k_1,M_2=k_2}-\vartheta\rVert\leq\frac{\eps}{2}.
\end{equation}
Given this, we now derive an upper bound on $I(Z^n\wedge M_0M_1M_2)$, where the random triple $(M_0,M_1,M_2)$ is uniformly distributed on the possible input message triples and $Z^n$ represents the output received by Eve. Observe that
\begin{align}
  &\mathrel{\hphantom{=}}I(Z^{n}\wedge M_0M_1M_2)\notag\\
  &=\frac{1}{K_0K_1K_2}\sum_{k_0,k_1,k_2}(H(Z^{n})-H(Z^{n}\vert M_0=k_0,M_1=k_1,M_2=k_2)).\label{darr}
\end{align}
By \cite[Lemma 2.7]{CK}, every summand on the right-hand side is upper-bounded by $\eps_{k_0k_1k_2}\log(\lvert\Z\rvert^n/\eps_{k_0k_1k_2})$ if
\[
  \eps_{k_0k_1k_2}:=\lVert P_{Z^n}-P_{Z^n\vert M_0=k_0,M_1=k_1,M_2=k_2}\rVert\leq\frac{1}{2}.
\]
But due to \eqref{totvarsmall}, 
\begin{align*}
  &\quad\;\lVert P_{Z^{n}}-P_{Z^{n}\vert M_0=k_0,M_1=k_1,M_2=k_2}\rVert\\
  &\leq\lVert P_{Z^{n}}-\vartheta\rVert+\lVert \vartheta-P_{Z^{n}\vert M_0=k_0,M_1=k_1,M_2=k_2}\rVert\\
  &\leq\frac{1}{K_0K_1K_2}\sum_{\tilde k_0,\tilde k_1,\tilde k_2}\lVert P_{Z^{n}\vert M_0=\tilde k_0,M_1=\tilde k_1,M_2=\tilde k_2}-\vartheta\rVert+\frac{\eps}{2}\\
  &\leq\eps.
\end{align*}
Thus if $\eps$ tends to zero exponentially in blocklength, then \eqref{darr} is upper-bounded by $\eps\log(\lvert\Z\rvert^n/\eps)$ which tends to zero in $n$.

\subsection{Probabilistic Bounds for Secrecy}\label{sect:setup}

In this subsection we define the random variables from which we will build a stochastic wiretap code in Subsection \ref{MAC}. For this family of random variables we prove several Chernoff-type estimates which will serve to find a code satisfying \eqref{totvarsmall}. For Case 3, two such estimates are sufficient, Case 0 and 2 require three and Case 1 requires four. Within each case, one deals with the joint typicality of the inputs at Alice\eins\ and Alice\zwei, and the other estimates base on each other. This is due to the complex structure of our family of random variables. Still, all the cases are nothing but a generalization of Devetak's approach taken in \cite{De05}. For each case, we first show the probabilistic bounds in one paragraph and then in another paragraph how to achieve \eqref{totvarsmall} from those bounds.

Let $p=P_U\otimes P_{X\vert U}\otimes P_{Y\vert U}\otimes W\in\Pi$, i.e. $p$ is the distribution of a random vector $(U,X,Y,T,Z)$. The auxiliary random variables $V_1$ and $V_2$ will be introduced later in the usual way of prefixing a channel as a means of additional randomization. Let $\delta>0$ and define for any $n$
\begin{align*}
  P_U^n(\u)&:=\frac{P_U^{\otimes n}(\u)}{P_U^{\otimes n}(T_{U,\delta}^n)}&(\u\in T_{U,\delta}^n),\\
  P_{X\vert U}^n(\x\vert\u)&:=\frac{P_{X\vert U}^{\otimes n}(\x\vert\u)}{P_{X\vert U}^{\otimes n}(T_{X\vert U,\delta}^n(\u)\vert\u)}&(\x\in T_{X\vert U,\delta}^n(\u),\u\in T_{U,\delta}^n),\\
  P_{Y\vert U}^n(\y\vert\u)&:=\frac{P_{Y\vert U}^{\otimes n}(\y\vert\u)}{P_{Y\vert U}^{\otimes n}(T_{Y\vert U,\delta}^n(\u)\vert\u)}&(\y\in T_{Y\vert U,\delta}^n(\u),\u\in T_{U,\delta}^n).
\end{align*}
Let $L_0,L_1,L_2$ be positive integers. We define $L_0$ independent families of random variables $(U^{l_0},\F_{l_0})$ as follows. $U^{l_0}$ is distributed according to $P_{U}^n$. We let $\F_{l_0}:=\{X^{l_0l_1},Y^{l_0l_2}:l_1\in[L_1],l_2\in[L_2]\}$ be a set of random variables which are independent given $U^{l_0}$ and which satisfy $X^{l_0l_1}\sim P_{X\vert U}^{n}(\,\cdot\,\vert U^{l_0})$ and $Y^{l_0l_2}\sim P_{Y\vert U}^{n}(\,\cdot\,\vert U^{l_0})$. Finally we define
\begin{equation}\label{ffamily}
  \F:=\bigcup_{l_0\in[L_0]}(U^{l_0},\F_{l_0}).
\end{equation}

Throughout the section, let a small $\eps>0$ be fixed. The core of the proofs of all the lemmas of this subsection is the following Chernoff bound, see e.g.\ \cite{APP}.

\begin{lemma}\label{Chernie}
  Let $b>0$ and $0<\eps<1/2$. For an independent sequence of random variables $Z_1,\ldots,Z_L$ with values in $[0,b]$ with $\mu_l:=\EE[X_l]$ and with $\mu:=\frac{1}{L}\sum_l\mu_l$ one has
\begin{align*}
  \PP\biggl[\frac{1}{L}\sum_{l=1}^LZ_l>(1+\eps)\mu\biggr]&\leq\exp\left(-L\cdot\frac{\eps^2\mu}{2b\ln 2}\right)\\
  \intertext{and}
  \PP\biggl[\frac{1}{L}\sum_{l=1}^LZ_l<(1-\eps)\mu\biggr]&\leq\exp\left(-L\cdot\frac{\eps^2\mu}{2b\ln 2}\right).
\end{align*}
\end{lemma}

In order to obtain useful bounds in the following we collect here some well-known estimates concerning typical sets, see e.g.\ \cite[Lemma 17.8]{CK}. Let $(A,B)$ be a random pair on the finite Cartesian product $\A\times\B$. Let $\xi,\zeta>0$. Then there exists a $\tilde c=\tilde c(\lvert\A\rvert\lvert\B\rvert)>0$ such that for sufficiently large $n$
\begin{align}
  P_{B\vert A}^{\otimes n}(T_{B\vert A,\zeta}^n(\a)^c\vert\a)&\leq2^{-n\tilde c\zeta^2}.\label{untyp2}
\end{align}
Further there is a $\tau=\tau(P_{AB},\xi,\zeta)$ with $\tau\rightarrow 0$ as $\xi,\zeta\rightarrow0$ such that 
\begin{align}
  P_{B\vert A}^{\otimes n}(\b\vert\a)&\leq2^{-n(H(B\vert A)-\tau)}&&\text{if }\a\in T_{A,\xi}^n,\b\in T_{B\vert A,\zeta}^n,\label{probbed}
\end{align}
and that for $n$ sufficiently large,
\begin{align}
  \lvert T_{A,\xi}^n\rvert&\leq 2^{n(H(A)+\tau)},\label{cardunbed}\\
  \lvert T_{B\vert A,\zeta}^n(\a)\rvert&\leq 2^{n(H(B\vert A)+\tau)}&&\text{if }\a\in T_{A,\xi}^n.\label{cardbed}
\end{align}
We set
\[
  c:=\tilde c(\lvert\U\rvert\lvert\X\rvert\lvert\Y\rvert\lvert\Z\rvert),
\]
this is the minimal $\tilde c$ we will need in the following.

\subsubsection{Bounds for Case 0 and 1:}

Let $L_0,L_1,L_2$ be arbitrary. Due to their conditional independence, the $X^{l_0l_1}$ and $Y^{l_0l_2}$ cannot be required to be jointly conditionally typical given $U^{l_0}$. However, the next lemma shows that most of them are jointly conditionally typical with high probability. 

\begin{lemma}\label{lemgemtyp3}
  For $(l_0,l_2)\in[L_0]\times[L_2]$, let the event $A\1_*(l_0,l_2)$ be defined by
\begin{multline*}
  A\1_*(l_0,l_2)\\:=\bigl\{\lvert\{l_1\in[L_1]:X^{l_0l_1}\in T_{X\vert YU,\delta}^{n}(Y^{l_0l_2},U^{l_0})\}\rvert\geq(1-\eps)(1-2\cdot2^{-{n}c\delta^2})L_1\}.
\end{multline*}
Then
\[
  \PP[A\1_*(l_0,l_2)^c]\leq\exp\left(-L_1\cdot\frac{\eps^2(1-2\cdot2^{-{n}c\delta^2})}{2\ln 2}\right).
\]
\end{lemma}

\begin{proof}
Let $\u\in T_{U,\delta}^n$ and $\y\in T_{Y\vert U,\delta}^n(\u)$. We first condition on the event $\{Y^{l_0l_2}=\y,U^{l_0}=\u\}$. Due to \eqref{untyp2}, we have
\begin{align*}
  &\PP[X^{11}\notin T_{X\vert YU,\delta}^{n}(\y,\u)\vert Y^{l_0l_2}=\y,U^1=\u]\\
     &=\frac{1}{P_{X\vert U}^{\otimes {n}}(T_{X\vert U,\delta}^{n}(\u)\vert\u)}\sum_{\x\in T_ {X\vert U,\delta}^{n}(\u)\setminus T_{X\vert YU,\delta}^{n}(\y,\u)}P_{X\vert U}^{\otimes {n}}(\x\vert\u)\\
     &\leq\frac{1}{P_{X\vert U}^{\otimes {n}}(T_{X\vert U,\delta}^{n}(\u)\vert\u)}\sum_{\x\notin T_{X\vert YU,\delta}^{n}(\y,\u)}P_{X\vert YU}^{\otimes {n}}(\x\vert\y,\u)\\
     &\leq\frac{2^{-{n}c\delta^2}}{1-2^{-{n}c\delta^2}}.
\end{align*}
In particular,
\begin{align*}
  \mu&:=\PP[X^{11}\in T_{X\vert YU,\delta}^{n}(\y,\u)\vert Y^{11}=\y,U^1=\u]
  \geq 1-2\cdot 2^{-{n}c\delta^2}.
\end{align*}
Therefore
\begin{multline*}
  \PP[A\1_*(l_0,l_2)^c\vert Y^{l_0l_2}=\y,U^{l_0}=\u]\\
  \leq\PP\biggl[\sum_{l_1}1_{T_{X\vert YU,\delta}^{n}(\y,\u)}(X^{l_0l_1})\leq(1-\eps)\mu\,L_1\biggl\vert Y^{l_0l_2}=\y,U^{l_0}=\u\biggr],
\end{multline*}
which by Lemma \ref{Chernie} can be bounded by
\[
  \exp\left(-L_1\cdot\frac{\eps^2\mu}{2\ln 2}\right)\leq\exp\left(-L_1\cdot\frac{\eps^2(1-2\cdot2^{-{n}c\delta^2})}{2\ln 2}\right).
\]
This completes the proof as this bound is independent of $(\y,\u)$.\qed
\end{proof}

Lemma \ref{lemgemtyp3} is not needed for a single sender. As we cannot guarantee the joint conditional typicality of both senders' inputs, we need to introduce an explicit bound on the channel transition probabilities. This is done in the set $E\1_1$. Then we prove three lemmas each of which exploits one of the three types of independence contained in $\F$. Altogether these lemmas provide lower bounds on $L_0,L_1,L_2$ which if satisfied allow the construction of a wiretap code satisfying \eqref{totvarsmall}. Let 
\[
  E\1_1(\u,\x,\y):=\{\z\in T_{Z\vert YU,2\lvert\X\rvert\delta}^{n}(\y,\u):W_e^{\otimes {n}}(\z\vert \x,\y)\leq2^{-{n}(H(Z\vert XY)-f_2(\delta))}\},
\]
where $f_2(\delta)=\tau(P_{UXYZ},3\delta,\delta)$ (see \eqref{probbed}). Let 
\[
  \vartheta\1_{\u\y}(\z):=\EE[W_e^{\otimes {n}}(\z\vert X^{11},\y)1_{E\1_1(\u,X^{11},\y)}(\z)\vert U^1=\u]
\]
and for
\[
  F\1_1(\u,\y):=\{\z\in T_{Z\vert YU,2\lvert\X\rvert\delta}^{n}(\y,\u):\vartheta\1_{\u\y}(\z)\geq \eps\lvert T_{Z\vert YU,2\lvert\X\rvert\delta}^{n}(\y,\u)\rvert^{-1}\}
\]
define
\[
  \hat\vartheta\1_{\u\y}:=\vartheta\1_{\u\y}\cdot1_{F\1_1(\u,\y)},\quad E\1_2(\u,\x,\y):=E\1_1(\u,\x,\y)\cap F\1_1(\u,\y).
\]

\begin{lemma}\label{AW3}
  For every $\z\in\Z^n$ and $(l_0,l_2)\in [L_0]\times[L_2]$, let $A\1_1(l_0,l_2,\z)$ be the event that
\[
  \frac{1}{L_1}\sum_{l_1}W_e^{\otimes {n}}(\z\vert X^{l_0l_1},Y^{l_0l_2})1_{E\1_2(U^{l_0},X^{l_0l_1},Y^{l_0l_2})}(\z)\in[(1\pm\eps)\hat\vartheta\1_{U^{l_0}Y^{l_0l_2}}(\z)].
\]
Then
\begin{align*}
  \PP[A\1_1(l_0,l_2,\z)^c]
  &\leq 2\exp\left(-L_1\cdot\frac{\eps^32^{-{n}(I(Z\wedge X\vert YU)+f_1(\delta)+f_2(\delta))}}{2\ln 2}\right)
\end{align*}
for $f_1(\delta)=\tau(P_{UYZ},2\delta,2\lvert\X\rvert\delta)$ and $n$ sufficiently large.
\end{lemma}
\begin{proof}
  For $\u\in T_{U,\delta}^n$ and $\y\in T_{Y\vert U,\delta}^n(\u)$ we condition on the event $\{Y^{l_0l_2}=\y,U^{l_0}=\u\}$. The conditional expectation of the bounded conditionally i.i.d.\ random variables 
\begin{align*}
  W_e^{\otimes {n}}(\z\vert X^{l_0l_1},\y)1_{E\1_2(\u,X^{l_0l_1},\y)}(\z)\leq2^{-{n}(H(Z\vert XY)-f_2(\delta))}&&(l_1\in[L_1])
\end{align*}
is $\hat\vartheta\1_{\u\y}(\z)$. We use Lemma \ref{Chernie}, the definition of $F\1_1(\u,\y)$, and \eqref{cardbed} to obtain for $n$ sufficiently large
\begin{align*}
  &\mathrel{\hphantom{\leq}}\PP[A\1_1(l_0,l_2,\z)^c\vert Y^{l_0l_2}=\y,U^{l_0}=\u]\\
  &\leq2\exp\left(-L_1\cdot\frac{\eps^2\hat\vartheta\1_{\u\y}(\z)2^{{n}(H(Z\vert XY)-f_2(\delta))}}{2\ln 2}\right)\\
  &\leq2\exp\left(-L_1\cdot\frac{\eps^32^{-{n}(I(Z\wedge X\vert YU)+f_1(\delta)+f_2(\delta))}}{2\ln 2}\right).
\end{align*}
This bound is uniform in $\u$ and $\y$, so the proof is complete.\qed
\end{proof}

For the next lemma, define
\[
  \vartheta\1_\u(\z):=\EE[W_e^{\otimes {n}}(\z\vert X^{11},Y^{11})1_{E\1_2(\u,X^{11},Y^{11})}(\z)\vert U^1=\u].
\]
Further let
\[
  F\1_2(\u):=\{\z\in T_{Z\vert U,3\lvert\Y\rvert\lvert\X\rvert\delta}^{n}(\u):\vartheta\1_\u(\z)\geq\eps\lvert T_{Z\vert U,3\lvert\Y\rvert\lvert\X\rvert\delta}^{n}(\u)\rvert^{-1}\}
\]
and
\[
  \hat\vartheta\1_\u=\vartheta\1_\u\cdot 1_{F\1_2(\u)},\quad E\1_0(\u,\x,\y):=E\1_2(\u,\x,\y)\cap F\1_2(\u,\y).
\]

\begin{lemma}\label{joint}
  For every $\z\in\Z^n$ and $l_0\in[L_0]$, let $A\1_2(l_0,\z)$ be the event
\[
  \frac{1}{L_1L_2}\sum_{l_1l_2}W_e^{\otimes {n}}(\z\vert X^{l_0l_1},Y^{l_0l_2})1_{E\1_0(U^{l_0},X^{l_0l_1},Y^{l_0l_2})}(\z)\in[(1\pm3\eps)\hat\vartheta\1_{U^{l_0}}(\z)].
\]
Then for $\eps$ sufficiently small and $n$ sufficiently large,
\begin{align*}
  \PP[A\1_2(l_0,\z)^c]
  &\leq 2\lvert\Y\rvert^{n}\exp\left(-L_1\cdot\frac{\eps^32^{-{n}(I(Z\wedge X\vert YU)+f_1(\delta)+f_2(\delta))}}{2\ln 2}\right)\\
  &\hphantom{\mathrel{\leq}}+2\exp\left(-L_2\cdot\frac{\eps^32^{-{n}(I(Z\wedge Y\vert U)+f_1(\delta)+f_4(\delta))}}{4\ln 2}\right),
\end{align*}
where $f_4(\delta)=\tau(P_{UZ},\delta,3\lvert\Y\rvert\lvert\X\rvert\delta)$.
\end{lemma}
\begin{proof}

Let $\u\in T_{U,\delta}^n$. We define the set $B_\u\subset(T_{X\vert U,\delta}^n(\u))^{L_1}$ as
\begin{multline*}
  \bigcap_{\y\in T_{Y\vert U,\delta}^n(\u)}\Bigl\{(\x^1,\ldots,\x^{L_1})\in(T_{X\vert U,\delta}^n(\u))^{L_1}:\\\frac{1}{L_1}\sum_{l_1}W_e^{\otimes {n}}(\z\vert \x^{l_0l_1},\y)1_{E\1_0(\u,X^{l_0l_1},\y)}(\z)\in[(1\pm\eps)\hat\vartheta\1_{\u\y}(\z)]\Bigr\}.
\end{multline*}
One has 
\begin{align*}
  &\PP[A\1_2(l_0,\z)^c\vert U^{l_0}=\u]\\
  &\leq\PP\bigl[\{(X^{l_01},\ldots,X^{l_0L_1})\notin B_\u\}\vert U^{l_0}=\u\bigr]\\
  &\mathrel{\hphantom{=}}+\sum_{(\x^1,\ldots,\x^{L_1})\in B_\u}\PP\bigl[A\1_2(l_0,\z)^c\vert X^{l_01}=\x^1,\ldots,X^{l_0L_1}=\x^{L_1},U^{l_0}=\u\bigr]\cdot\\
  &\hspace{.3\columnwidth}\cdot\PP[X^{l_01}=\x^1,\ldots,X^{l_0L_1}=\x^{L_1}\vert U^{l_0}=\u\bigr].
\end{align*}
From the proof of Lemma \ref{AW3} it follows that 
\begin{multline}\label{badx}
  \PP\bigl[\{(X^{l_01},\ldots,X^{l_0L_1})\notin B_\u\}\vert U^{l_0}=\u\bigr]\\\leq 2\lvert\Y\rvert^{n}\exp\left(-L_1\cdot\frac{\eps^32^{-{n}(I(Z\wedge X\vert YU)+f_1(\delta)+f_2(\delta))}}{2\ln 2}\right),
\end{multline}
which gives a bound independent of $\u$. Now let $(\x^1,\ldots,\x^{L_1})\in B_\u$. By \eqref{untyp2} and \eqref{probbed},
\begin{align*}
  \hat\vartheta\1_{\u\y}(\z)&=\EE[W_e^{\otimes {n}}(\z\vert X^{11},\y)1_{E\1_2(\u,X^{11},\y)}(\z)\vert U^1=\u]\\
  &\leq \EE[W_e^{\otimes {n}}(\z\vert X^{11},\y)\vert U^1=\u]\\
  &\leq \frac{1}{P_{X\vert U}^{\otimes {n}}(T_{X\vert U,\delta}^{n}(\u)\vert\u)}(P_{Z\vert YU})^{\otimes {n}}(\z\vert\y,\u)\\
  &\leq(1-2^{-{n}c\delta^2})^{-1}2^{-{n}(H(Z\vert YU)-f_1(\delta))}.
\end{align*}
Hence the random variables 
\begin{align*}
  \tilde W\1_{\u\z}(l_0,l_2):=\frac{1}{L_1}\sum_{l_1}W_e^{\otimes {n}}(\z\vert\x^{l_1},Y^{l_0l_2})1_{E\1_0(\u,\x^{l_1},Y^{l_0l_2})}(\z)&&(l_2\in[L_2]),
\end{align*}
which are independent conditional on $\{U^{l_0}=\u\}$, are upper-bounded by 
\[
  \frac{(1+\eps)}{(1-2^{-{n}c\delta^2})}\cdot2^{-{n}(H(Z\vert YU)-f_1(\delta))}.
\]
For their conditional expectation we have 
\begin{multline*}
  \mu_{l_0l_2}:=\EE[\tilde W\1_{\u\z}(l_0,l_2)\vert U^{l_0}=\u]\\\in[(1\pm\eps)\EE[\hat\vartheta\1_{\u Y^{l_0l_2}}(\z)\vert U^1=\u]]=[(1\pm\eps)\hat\vartheta\1_\u(\z)].
\end{multline*}
Thus their arithmetic mean $\bar\mu=(1/L_2)\sum_{l_2}\mu_{l_0l_2}$ must also be contained in $[(1\pm\eps)\hat\vartheta\1_\u(\z)]$. Applying Lemma \ref{Chernie}, we conclude
\begin{align*}
  &\mathrel{\hphantom{=}}\PP\bigl[A\1_2(l_0,\z)^c\vert X^{l_01}=\x^1,\ldots,X^{l_0L_1}=\x^{L_1},U^{l_0}=\u\bigr]\\
  &=\PP\biggl[\frac{1}{L_2}\sum_{l_2}\tilde W\1_{\u\z}(l_0,l_2)\notin[(1\pm3\eps)\hat\vartheta\1_\u(\z)]\biggl\vert U^{l_0}=\u\biggr]\\
  &\leq\PP\biggl[\frac{1}{L_2}\sum_{l_2}\tilde W\1_{\u\z}(l_0,l_2)\notin[(1\pm\eps)\bar\mu]\biggl\vert U^{l_0}=\u\biggr]\\
  &\leq2\exp\left(-L_2\cdot\frac{\eps^2(1-2^{-{n}c\delta^2})2^{{n}(H(Z\vert YU)-f_1(\delta))}(1-\eps)\hat\vartheta\1_{\u}(\z)}{2(1+\eps)\ln 2}\right).
\end{align*}
  Due to the definition of $F\1_2(\u)$ and to \eqref{cardbed}, this is smaller than
\begin{equation}\label{goodx}
  2\exp\left(-L_2\cdot\frac{\eps^32^{-{n}(I(Z\wedge Y\vert U)+f_1(\delta)+f_4(\delta))}}{4\ln 2}\right)
\end{equation}
if $\eps$ is sufficiently small and $n$ is sufficiently large, giving a bound independent of $\u$ and $\x^1,\ldots,\x^{L_1}$. Adding the bounds \eqref{badx} and \eqref{goodx} concludes the proof.\qed
\end{proof}

The next lemma is only needed in Case 1. Let $A\1_2(\z):=A\1_2(1,\z)\cap\ldots\cap A\1_2(L_0,\z)$. For every $\z$, we then define a new probability measure by $\hat\PP\1_\z:=\PP[\cdot\vert A\1_2(\z)]$. With $\vartheta\1(\z):=\hat\EE\1_\z[\hat\vartheta\1_{U^1}(\z)]$ define
\[
  F\1_0:=\{\z\in T_{Z,4\lvert\Y\rvert\lvert\X\rvert\lvert\U\rvert\delta}^{n}:\vartheta\1(\z)\geq\lvert T_{Z,4\lvert\Y\rvert\lvert\X\rvert\lvert\U\rvert\delta}^{n}\rvert^{-1}\}
\]
and $\hat\vartheta\1:=\vartheta\1\cdot1_{F\1_0}$.

\begin{lemma}\label{theucase}
  Let $\z\in F\1_0$ and let $A\1_0(\z)$ be the event that
\[
  \frac{1}{L_0L_1L_2}\sum_{l_0,l_1,l_2}W_e^{\otimes {n}}(\z\vert X^{l_0l_1},Y^{l_0l_2})1_{E\1_0(U^{l_0},X^{l_0l_1},Y^{l_0l_2})}(\z)\in[(1\pm5\eps)\hat\vartheta\1(\z)].
\]
Then for $f_6(\delta)=\tau(P_{Z},4\lvert\Y\rvert\lvert\X\rvert\lvert\U\rvert\delta,\delta)$, sufficiently small $\eps$ and $n$ sufficiently large,
\begin{align*}
  &\mathrel{\hphantom{\leq}}\PP[A\1_0(\z)^c]\\
  &\leq 2L_0\lvert\Y\rvert^{n}\exp\left(-L_1\cdot\frac{\eps^32^{-{n}(I(Z\wedge X\vert YU)+f_1(\delta)+f_2(\delta))}}{2\ln 2}\right)\\
  &\mathrel{\hphantom{\leq}}+2L_0\exp\left(-L_2\cdot\frac{\eps^32^{-{n}(I(Z\wedge Y\vert U)+f_1(\delta)+f_4(\delta))}}{4\ln 2}\right)\\
  &\mathrel{\hphantom{\leq}}+2\exp\left(-L_0\cdot\frac{\eps^32^{-{n}(I(Z\wedge U)+f_4(\delta)+f_6(\delta))}}{4\ln 2}\right).
\end{align*}
\end{lemma}

\begin{proof}
We have
\begin{align}
  \PP[A\1_0(\z)^c]
  &\leq\hat\PP\1_\z[A\1_0(\z)^c]+\PP[A\1_2(\z)^c].\label{totwahr3}
\end{align}
By Lemma \ref{joint}, for $\eps$ sufficiently small and $n$ sufficiently large,
\begin{equation}\label{by}
\begin{aligned}
  \PP[A\1_2(\z)^c]
  &\leq 2L_0\lvert\Y\rvert^{n}\exp\left(-L_1\cdot\frac{\eps^32^{-{n}(I(Z\wedge X\vert YU)+f_1(\delta)+f_2(\delta))}}{2\ln 2}\right)\\
  &\mathrel{\hphantom{\leq}}+2L_0\exp\left(-L_2\cdot\frac{\eps^32^{-{n}(I(Z\wedge Y\vert U)+f_1(\delta)+f_4(\delta))}}{4\ln 2}\right).
\end{aligned}
\end{equation}
In order to bound $\hat\PP\1_\z[A\1_0(\z)^c]$, note that the sets $A\1_2(1,\z),\ldots,A\1_2(L_0,\z)$ are independent with respect to $\PP$. Thus under $\hat\PP\1_\z$, the random variables 
\begin{align*}
  \tilde W\1_\z(l_0):=\frac{1}{L_1L_2}\sum_{l_1,l_2}W_e^{\otimes {n}}(\z\vert X^{l_0l_1},Y^{l_0l_2})1_{E\1_0(U^{l_0},X^{l_0l_1},Y^{l_0l_2})}(\z)&&(l_0\in[L_0])
\end{align*}
retain their independence and are upper-bounded by 
\[
  (1+3\eps)\max_{\u\in T_{U,\delta}^n}\hat\vartheta\1_{\u}(\z).
\]
We can further bound this last term as follows: for $\u\in T_{U,\delta}^n$, applying \eqref{untyp2} and \eqref{probbed}, 
\begin{align*}
  \hat\vartheta\1_\u(\z)&=\EE[W_e^{\otimes {n}}(\z\vert X^{11},Y^{11})1_{E\1_0(\u,X^{11},Y^{11})}(\z)\vert U^1=\u]\\
  &\leq\EE[W_e^{\otimes {n}}(\z\vert X^{11},Y^{11})\vert U^1=\u]\\
  &\leq\frac{1}{P_1^{\otimes {n}}(T_{X\vert U,\delta}^{n}(\u)\vert\u)P_2^{\otimes {n}}(T_{Y\vert U,\delta}^{n}(\u)\vert\u)}
	  P_{Z\vert U}^{\otimes n}(\z\vert\u)\\
  &\leq(1-2^{-{n}c_1\delta^2})^{-2}2^{-{n}(H(Z\vert U)-f_4(\delta))}.
\end{align*}
Observing that $\hat\EE\1_\z[\tilde W\1_\z(1)]\in[(1\pm3\eps)\hat\vartheta\1(\z)]$ and applying Lemma \ref{Chernie} and \eqref{cardbed} in the usual way  yields
\begin{align*}
  \hat\PP\1_\z[A\1_0(\z)^c]&\leq2\exp\left(-L_0\cdot\frac{\eps^2(1-2^{-nc\delta^2})^2\,2^{n(H(Z\vert U)-f_4(\delta))}(1-3\eps)\,\hat\vartheta\1(\z)}{2(1+3\eps)\ln 2}\right)\\
  &\leq2\exp\left(-L_0\cdot\frac{\eps^32^{-n(I(Z\wedge U)+f_4(\delta)+f_6(\delta))}}{4\ln 2}\right)
\end{align*}
if $\eps$ is sufficiently small and $n$ sufficiently large. Inserting this and \eqref{by} in \eqref{totwahr3} completes the proof.\qed
\end{proof}

We finally note that results analogous to Lemma \ref{lemgemtyp3}-\ref{theucase} hold where the roles of $X$ and $Y$ are exchanged. We denote the corresponding events by $A\1_*(l_0,l_2)'$ and $A\1_1(l_0,l_2,\z)',A\1_2(l_0,\z)',A\1_0(\z)'$.

\subsubsection{Secrecy for Case 0 and 1:} The following lemma links the above probabilistic bounds to secrecy. In the next subsection, roughly speaking, we will associate a family $\F$ to every message triple $(k_0,k_1,k_2)$. If $L_0,L_1,L_2$ are large enough, the bounds of Lemma \ref{sec3} are satisfied for every such $\F$ with high probability. Hence there is a joint realization of the $\F$ such that the statement of the lemma is satisfied for every message triple. By an appropriate choice of random code one then obtains \eqref{totvarsmall}.

\begin{lemma}\label{sec3}
  Denote by $p\1$ the bound on $\PP[A\1_2(l_0,\z)^c]$ derived in Lemma \ref{joint}. Let $\{\u^{l_0},\x^{l_0l_1},\y^{l_0l_2}:(l_0,l_1,l_2)\in[L_0]\times[L_1]\times[L_2]\}$ be a realization of $\F$ satisfying the conditions of
\begin{align}
  \bigcap_{l_0,l_2}&A\1_*(l_0,l_2),\label{cond3,*}\\
  \bigcap_{l_0,l_2}\bigcap_{\z\in\Z^n}&A\1_1(l_0,l_2,\z),\label{cond3,1}\\
  \bigcap_{l_0}\bigcap_{\z\in\Z^n}&A\1_2(l_0,\z),\label{cond3,2}\\
  \bigcap_{\z\in F\1_0}&A\1_0(\z).\label{cond3,3}
\end{align}
Then
\[
  \lVert\hat\vartheta\1-\frac{1}{L_0L_1L_2}\sum_{l_0,l_1,l_2}W_e^{\otimes n}(\cdot\vert\x^{l_0l_1},\y^{l_0l_2})\rVert\leq 20\eps+9\cdot2^{-nc\delta^2}+L_0\lvert\Z\rvert^np\1.
\]
The same inequality is true if we require conditions \textnormal{(\ref{cond3,*}$'$)}-\textnormal{(\ref{cond3,3}$'$)} which contain the primed equivalents of \eqref{cond3,*}-\eqref{cond3,3} defined at the end of the previous paragraph. If $L_0=1$, then \eqref{cond3,3} and \textnormal{(\ref{cond3,3}$'$)} do not have to hold.
\end{lemma}

We now prove the above lemma. We have
\begin{align}
  &\quad\;\lVert\hat\vartheta\1-\frac{1}{L_0L_1L_2}\sum_{l_0,l_1,l_2}W_e^{\otimes n}(\cdot\vert\x^{l_0l_1},\y^{l_0l_2})\rVert\notag\\
  &\leq\lVert\hat\vartheta\1-\frac{1}{L_0L_1L_2}\sum_{l_0,l_1,l_2}W_e^{\otimes n}(\cdot\vert\x^{l_0l_1},\y^{l_0l_2})1_{E\1_0(\u^{l_0},\x^{l_0l_1},\y^{l_0l_2})}1_{F\1_0}\rVert\label{3I}\\
  &+\lVert\frac{1}{L_0L_1L_2}\sum_{l_0,l_1,l_2}W_e^{\otimes n}(\cdot\vert\x^{l_0l_1},\y^{l_0l_2})1_{E\1_0(\u^{l_0},\x^{l_0l_1},\y^{l_0l_2})}(1-1_{F\1_0})\rVert\label{3II}\\
  &+\lVert\frac{1}{L_0L_1L_2}\sum_{l_0,l_1,l_2}W_e^{\otimes n}(\cdot\vert\x^{l_0l_1},\y^{l_0l_2})1_{E\1_2(\u^{l_0},\x^{l_0l_1},\y^{l_0l_2})}(1-1_{F\1_2(\u^{l_0})})\rVert\label{3III}\\
  &+\lVert\frac{1}{L_0L_1L_2}\sum_{l_0,l_1,l_2}W_e^{\otimes n}(\cdot\vert\x^{l_0l_1},\y^{l_0l_2})1_{E\1_1(\u^{l_0},\x^{l_0l_1},\y^{l_0l_2})}(1-1_{F\1_1(\u^{l_0},\y^{l_0l_2})})\rVert\label{3IV}\\
  &+\lVert\frac{1}{L_0L_1L_2}\sum_{l_0,l_1,l_2}W_e^{\otimes n}(\cdot\vert\x^{l_0l_1},\y^{l_0l_2})(1-1_{E\1_1(\u^{l_0},\x^{l_0l_1},\y^{l_0l_2})})\rVert.\label{3V}
\end{align}
Due to \eqref{cond3,3}, we know that $\eqref{3I}\leq5\eps$.

Next we consider \eqref{3IV}. Due to \eqref{cond3,1} we have
\begin{align*}
  &\mathrel{\hphantom{=}}\eqref{3IV}\\
  &\leq1-\frac{1}{L_0L_1L_2}\sum_{l_0,l_1,l_2}W_e^{\otimes n}(E\1_2(\u^{l_0},\x^{l_0l_1},\y^{l_0l_2})\vert\x^{l_0l_1},\y^{l_0l_2})\\
  &\leq1-\frac{1-\eps}{L_0L_2}\sum_{l_0,l_2}\hat\vartheta\1_{\u^{l_0}\y^{l_0l_2}}(\Z^n)
\end{align*}
(we defined the general measure of a set in the notation section at the beginning of the paper). The support of $\vartheta\1_{\u^{l_0}\y^{l_0l_2}}$ is contained in $T_{Z\vert YU,2\lvert\X\rvert\delta}^n(\y^{l_0l_2},\u^{l_0})$, so by the definition of $F\1_1(\u^{l_0},\y^{l_0l_2})$ we obtain
\begin{equation}\label{thetabound}
  \hat\vartheta\1_{\u^{l_0}\y^{l_0l_2}}(\Z^n)\geq\vartheta\1_{\u^{l_0}\y^{l_0l_2}}(\Z^n)-\eps.
\end{equation}

\begin{lemma}\label{lemabschlemma}
  If $\u\in T_{U,\delta}^n$ and $\y\in T_{Y\vert U,\delta}^n(\u)$, then 
\[
  \vartheta\1_{\u\y}(\Z^n)\geq1-2\cdot2^{-nc\delta^2}.
\]
\end{lemma}

\begin{proof}
  First of all note that
\begin{align}
  &\mathrel{\hphantom{=}}\vartheta\1_{\u\y}(\Z^n)\notag\\
  &=\EE[W_e^{\otimes n}(E\1_1(\u,X^{11},\y)\vert X^{11},\y)\vert U^1=\u]\notag\\
  &\geq\EE[W_e^{\otimes n}(E\1_1(\u,X^{11},\y)\vert X^{11},\y);X^{11}\in T_{X\vert YU,\delta}^n(\y,\u)\vert U^1=\u].\label{rhs}
\end{align}
Now we claim that for $\x\in T_{X\vert YU,\delta}^n(\y,\u)$ 
\begin{equation}\label{claim3}
  T_{Z\vert YXU,\delta}^n(\y,\x,\u)\subset T_{Z\vert YU,2\lvert\X\rvert\delta}^n(\y,\u).
\end{equation}
To verify this, let $(z,y,u)\in\Z\times\Y\times\U$ and $\z\in T_{Z\vert YXU,\delta}^n(\y,\x,\u)$. Then
\begin{align*}
  &\quad\;\left\lvert\frac{1}{{n}}N(z,y,u\vert\z,\y,\u)-P_{Z\vert YU}(z\vert y,u)\frac{1}{n}N(y,u\vert\y,\u)\right\rvert\\
  &\leq\sum_{x}\left\lvert\frac{1}{{n}}N(z,y,x,u\vert\z,\y,\x,\u)-W(z\vert x,y)\frac{1}{n}N(y,x,u\vert\y,\x,\u)\right\rvert\\
  &\mathrel{\hphantom{\leq}}+\sum_{x}W(z\vert x,y)\left\lvert\frac{1}{n}N(y,x,u\vert\y,\x,\u)-P_{X\vert YU}(x\vert y,u)\frac{1}{n}N(y,u\vert\y,\u)\right\rvert\\
  &\leq2\lvert\X\rvert\delta.
\end{align*}
This proves \eqref{claim3}. Due to the choice of $f_2(\delta)$ and to \eqref{probbed}, we thus see that $T_{Z\vert YXU,\delta}^n(\y,\x,\u)$ is contained in $E\1_1(\u,\x,\y)$ for $\x\in T_{X\vert YU,\delta}^n(\y,\u)$, and we have that \eqref{rhs} is lower-bounded by 
\begin{equation}\label{zwischenequation3}
  \EE[W_e^{\otimes n}(T_{Z\vert YXU,\delta}^n(\y,X^{11},\u)\vert X^{11},\y);X^{11}\in T_{X\vert YU,\delta}^n(\y,\u)\vert U^1=\u].
\end{equation}
Further, as in the proof of Lemma \ref{lemgemtyp3} one sees that
\begin{equation}\label{asin3}
  \PP[X^{11}\in T_{X\vert YU,\delta}^n(\y,\u)\vert U^1=\u]\geq1-\frac{2^{-nc\delta^2}}{1-2^{-nc\delta^2}}.
\end{equation} 
Due to \eqref{asin3} and \eqref{untyp2}, we can lower-bound \eqref{zwischenequation3} for sufficiently large $n$ by 
\[
  (1-2^{-nc\delta^2})\cdot\left(1-\frac{2^{-nc\delta^2}}{1-2^{-nc\delta^2}}\right)\geq 1-2\cdot2^{-nc\delta^2},
\]
which proves Lemma \ref{lemabschlemma}.\qed
\end{proof}
Using \eqref{thetabound} and Lemma \ref{lemabschlemma} we can conclude that
\[
  \eqref{3IV}\leq2(\eps+2^{-nc\delta^2}).
\]

One starts similarly for \eqref{3III}. We have by \eqref{cond3,2}
\begin{align*}
  \eqref{3III}&\leq1-\frac{1}{L_0L_1L_2}\sum_{l_0,l_1,l_2}W_e^{\otimes n}(E\1_0(\u^{l_0},\x^{l_0l_1},\y^{l_0l_2})\vert\x^{l_0l_1},\y^{l_0l_2})\\
  &\leq1-\frac{(1-3\eps)}{L_0}\sum_{l_0}\hat\vartheta\1_{\u^{l_0}}(\Z^n).
\end{align*}
As the support of $\vartheta\1_{\u^{l_0}}$ is contained in $T_{Z\vert U,3\lvert\Y\rvert\lvert\X\rvert\delta}^n(\u^{l_0})$, we can lower-bound $\hat\vartheta\1_{\u^{l_0}}(\Z^n)$ by $\vartheta\1_{\u^{l_0}}(\Z^n)-\eps$. Using \eqref{thetabound} and Lemma \ref{lemabschlemma}, we have
\begin{equation}\label{thetabound1}
  \vartheta\1_{\u^{l_0}}(\Z^n)=\EE[\hat\vartheta\1_{\u^{l_0}Y^{11}}(\Z^n)\vert U^1=\u^{l_0}]\geq1-2\cdot2^{-nc\delta^2}-\eps,
\end{equation}
so we conclude
\[
  \eqref{3III}\leq5\eps+2\cdot2^{-nc\delta^2}.
\]

For \eqref{3II}, one has by \eqref{cond3,3}
\begin{align*}
  \eqref{3II}&\leq1-\frac{1}{L_0L_1L_2}\sum_{l_0,l_1,l_2}W_e^{\otimes n}(E\1_0(\u^{l_0},\x^{l_0l_1},\y^{l_0l_2})\cap F\1_0\vert\x^{l_0l_1},\y^{l_0l_2})\\
  &\leq1-(1-5\eps)\hat\vartheta\1(F\1_0).
\end{align*}
It remains to lower-bound $\hat\vartheta\1(F\1_0)$. Observe that the support of $\vartheta\1$ is restricted to $T_{Z,4\lvert\Y\rvert\lvert\X\rvert\lvert\U\rvert\delta}^n$, so due to the definition of $F\1_0$, one has $\hat\vartheta\1(F\1_0)=\vartheta\1(F\1_0)\geq\vartheta\1(\Z^n)-\eps$. Further,
\begin{align*}
  \vartheta\1(\Z^n)&=\sum_{\z\in\Z^n}\hat\EE\1_\z[\hat\theta\1_{U^1}(\z)]\\
  &\geq\EE[\theta\1_{U^1}(\Z^n)]-\sum_{\z\in\Z^n}\PP[A\1_2(\z)^c]\\
  &=\EE[\theta\1_{U^1}(\Z^n)]-L_0\lvert\Z\rvert^np\1.
\end{align*}
In \eqref{thetabound1}, the integrand of $\EE[\theta\1_{U^1}(\Z^n)]$ was lower-bounded by $1-2\cdot2^{-nc\delta^2}-\eps$. We conclude
\[
  \eqref{3II}\leq7\eps+2\cdot2^{-nc\delta^2}+L_0\lvert\Z\rvert^np\1.
\]

Finally, we use condition \eqref{cond3,*} to bound \eqref{3V}. We have
\begin{align}
  &\mathrel{\hphantom{=}}\eqref{3V}\\
  &=\frac{1}{L_0L_1L_2}\sum_{l_0,l_1,l_2}W_e^{\otimes n}(E\1_1(\u^{l_0},\x^{l_0l_1},\y^{l_0l_2})^c\vert \x^{l_0l_1},\y^{l_0l_2})\notag\\
  &=\frac{1}{L_0L_2}\sum_{l_0,l_2}\biggl(\notag\\
  &\quad\frac{1}{L_1}\sum_{l_1:\x^{l_0l_1}\in T_{X\vert YU,\delta}^n(\y^{l_0l_2},\u^{l_0})}W_e^{\otimes n}(E\1_1(\u^{l_0},\x^{l_0l_1},\y^{l_0l_2})^c\vert \x^{l_0l_1},\y^{l_0l_2})\label{3l_0sum1}\\
  &+\frac{1}{L_1}\sum_{l_1:\x^{l_0l_1}\notin T_{X\vert YU,\delta}^n(\y^{l_0l_2},\u^{l_0})}W_e^{\otimes n}(E\1_1(\u^{l_0},\x^{l_0l_1},\y^{l_0l_2})^c\vert \x^{l_0l_1},\y^{l_0l_2})\biggr).\label{3l_0sum2}
\end{align}
For every $(l_0,l_2)$, we use $T_{Z\vert YXU,\delta}^n(\y,\x,\u)\subset E\1_1(\u,\x,\y)$ for $(\u,\x,\y)\in T_{U,\delta}^n\times T_{Y\vert U,\delta}^n(\u)\times T_{X\vert YU,\delta}^n(\y,\u)$ as shown in the proof of Lemma \ref{lemabschlemma} to upper-bound the term in \eqref{3l_0sum1} by $2^{-nc\delta^2}$. For \eqref{3l_0sum2}, we know from assumption \eqref{cond3,*} that it is at most $1-(1-\eps)(1-2\cdot2^{-nc\delta^2})$. Thus
\[
  \eqref{3V}\leq 2^{-nc\delta^2}+(1-\eps)(1-2\cdot2^{-nc\delta^2})\leq\eps+3\cdot2^{-nc\delta^2}.
\]

Collecting the bounds on \eqref{3I}-\eqref{3V}, we obtain a total upper bound of
\[
  20eps+9\cdot2^{-nc\delta^2}+L_0\lvert\Z\rvert^np\1.
\]
This finishes the proof of Lemma \ref{sec3}.

\subsubsection{Bounds for Case 2:}

Now we specialize to the case that $L_2=1$, but $L_0$ and $L_1$ arbitrary. This reduces the number of Chernoff-type estimates needed by one. Lemma \ref{AW3} carries over, Lemma \ref{joint} is not needed, but Lemma \ref{theucase} changes. We write $Y^{l_01}=:Y^{l_0}$. The definitions of $E\1_1(\u,\x,\y),F\1_1(\u,\y)$ and $\vartheta\1_{\u\y}$ carry over to this case, we just call them $E\2_1(\u,\x,\y),F\2_1(\u,\y)$ and $\vartheta\2_{\u\y}$. Further we define 
\[
  E\2_0(\u,\x,\y):=E\2_1(\u,\x,\y)\cap F\2_1(\u,\y).
\]
For every $l_0$, let $A\2_1(l_0,\z):=A\1_1(l_0,1,\z)$ and we set $A\2_1(\z):=A\2_1(1,\z)\cap\ldots\cap A\2_1(L_0,\z)$. We define for every $\z$ a new probability measure by $\hat\PP\2_\z:=\PP[\cdot\vert A\2_1(\z)]$. Let 
\[
  \vartheta\2(\z):=\hat\EE\2_\z[\hat\vartheta\2_{U^1Y^1}(\z)].
\]
Further let
\[
  F\2_0:=\{\z\in T_{Z,4\lvert\Y\rvert\lvert\X\rvert\lvert\U\rvert\delta}^{n}:\vartheta\2(\z)\geq\eps\lvert T_{Z,\delta}^{n}\rvert^{-1}\}
\]
and
\[
  \hat\vartheta\2=\vartheta\2\cdot 1_{F\2_0}.
\]

\begin{lemma}\label{joint2}
  Let $\z\in F\2_0$. Let $A\2_0(\z)$ be the event
\[
  \frac{1}{L_0L_1}\sum_{l_0,l_1}W_e^{\otimes {n}}(\z\vert X^{l_0l_1},Y^{l_0})1_{E\2_0(U^{l_0},X^{l_0l_1},Y^{l_0})}(\z)\in[(1\pm3\eps)\hat\vartheta\2(\z)].
\]
Then for $\eps$ sufficiently small and $n$ sufficiently large,
\begin{align*}
  \PP[A\2_0(\z)^c]&\leq 2L_0\exp\left(-L_1\cdot\frac{\eps^32^{-{n}(I(Z\wedge X\vert YU)+f_1(\delta)+f_2(\delta))}}{2\ln 2}\right)\\
  &\hphantom{\mathrel{\leq}}+2\exp\left(-L_0\cdot\frac{\eps^32^{-{n}(I(Z\wedge YU)+f_1(\delta)+f_6(\delta))}}{4\ln 2}\right).
\end{align*}
\end{lemma}
\begin{proof}
We have
\begin{equation}\label{totwahr}
  \PP[A\2_0(\z)^c]\leq\hat\PP\2_\z[A\2_0(\z)^c]+\PP[A\2_1(\z)^c].
\end{equation}
By Lemma \ref{AW3}, we know that
\begin{equation}\label{theucaseabsch1}
  \PP[A\2_1(\z)^c]\leq2L_0\exp\left(-L_1\cdot\frac{\eps^32^{-{n}(I(Z\wedge X\vert YU)+f_1(\delta)+f_2(\delta))}}{2\ln 2}\right).
\end{equation}
In order to bound $\PP\2_\z[A\2_0(\z)]$, note that the sets $A\2_1(1,\z),\ldots,A\2_1(L_0,\z)$ are independent with respect to $\PP$. Thus under $\hat\PP\2_\z$, the random variables 
\begin{align*}
  \tilde W\2_\z(l_0):=\frac{1}{L_1}\sum_{l_1}W_e^{\otimes {n}}(\z\vert X^{l_0l_1},Y^{l_0})1_{E\2_0(U^{l_0},X^{l_0l_1},Y^{l_0})}(\z)&&(l_0\in[L_0])
\end{align*}
retain their independence and are upper-bounded by 
\[
  (1+\eps)\max_{\u\in T_{U,\delta}^n}\max_{\y\in T_{Y\vert U,\delta}^n(\u)}\hat\vartheta\2_{\u\y}(\z).
\]
We can further bound this last term as follows: for $\u\in T_{U,\delta}^n$ and $\y\in T_{Y\vert U,\delta}^n(\u)$ one obtains by \eqref{untyp2} and \eqref{probbed}
\begin{align*}
  \hat\vartheta\2_{\u\y}(\z)&\leq\EE[W_e^{\otimes n}(\z\vert X^{11},\y)\vert U^{l_0}=\u]\\
  &\leq \frac{1}{1-2^{-nc\delta^2}}P^{\otimes n}_{Z\vert YU}(\z\vert\y,\u)\\
  &\leq\frac{1}{1-2^{-nc\delta^2}}2^{-n(H(Z\vert YU)-f_1(\delta))}.
\end{align*}
Observing that $\hat\EE_\z[\tilde W\2_\z(1)]\in[(1\pm\eps)\hat\vartheta\2(\z)]$ and applying Lemma \ref{Chernie} in the usual way  yields
\begin{align*}
  \hat\PP_\z[A\2_0(\z)]&\leq2\exp\left(-L_0\cdot\frac{\eps^2(1-2^{-nc\delta^2})2^{n(H(Z\vert YU)-f_1(\delta))}(1-\eps)\hat\vartheta\2(\z)}{2(1+\eps)\ln 2}\right)\\
  &\leq2\exp\left(-L_0\cdot\frac{\eps^32^{-n(I(Z\wedge YU)+f_1(\delta)+f_6(\delta))}}{4\ln 2}\right)
\end{align*}
if $\eps$ is sufficiently small and $n$ sufficiently large. Inserting this and \eqref{theucaseabsch1} in \eqref{totwahr} completes the proof.\qed
\end{proof}

Again we note that a result analogous to Lemma \ref{joint2} holds where the roles of $X$ and $Y$ are exchanged. Setting $A\2_*(l_0):=A\1_*(l_0,1)$, we denote the events corresponding to such an exchange by $A\2_*(l_0)'$ and $A\2_1(l_0,\z)',A\2_0(\z)'$.

\subsubsection{Secrecy for Case 2:}

\begin{lemma}\label{sec2}
  Denote by $p\2$ the bound on $\PP[A\2_1(l_0,\z)^c]$ derived in Lemma \ref{AW3}. Let $\{\u^{l_0},\x^{l_0l_1},\y^{l_0}:(l_0,l_1,l_2)\in[L_0]\times[L_1]\times[L_2]\}$ be a realization of $\F$ satisfying the conditions of
\begin{align}
  \bigcap_{l_0}&A\2_*(l_0),\label{cond2,*}\\
  \bigcap_{l_0}\bigcap_{\z\in\Z^n}&A\2_1(l_0,\z),\label{cond2,1}\\
  \bigcap_{\z\in F\2_0}&A\2_0(\z).\label{cond2,2}
\end{align}
Then
\[
  \lVert\hat\vartheta\2-\frac{1}{L_0L_1}\sum_{l_0,l_1}W_e^{\otimes n}(\cdot\vert\x^{l_0l_1},\y^{l_0})\rVert\leq9\eps+7\cdot2^{-nc\delta^2}+L_0\lvert\Z\rvert^np\2.
\]
The same inequality is true if we require conditions \textnormal{(\ref{cond2,*}$'$)}-\textnormal{(\ref{cond2,2}$'$)} which contain the primed equivalents of \eqref{cond2,*}-\eqref{cond2,2} defined at the end of the previous paragraph.
\end{lemma}

We use this subsection to prove the above lemma. We have
\begin{align}
  &\quad\;\lVert\hat\vartheta\2-\frac{1}{L_0L_1}\sum_{l_0,l_1}W_e^{\otimes n}(\cdot\vert\x^{l_0l_1},\y^{l_0})\rVert\notag\\
  &\leq\lVert\hat\vartheta\2-\frac{1}{L_0L_1}\sum_{l_0,l_1}W_e^{\otimes n}(\cdot\vert\x^{l_0l_1},\y^{l_0})1_{E\2_0(\u^{l_0},\x^{l_0l_1},\y^{l_0})}1_{F\2_0}\rVert\label{2I}\\
  &+\lVert\frac{1}{L_0L_1}\sum_{l_0,l_1}W_e^{\otimes n}(\cdot\vert\x^{l_0l_1},\y^{l_0})1_{E\2_0(\u^{l_0},\x^{l_0l_1},\y^{l_0})}(1-1_{F\2_0})\rVert\label{2II}\\
  &+\lVert\frac{1}{L_0L_1}\sum_{l_0,l_1}W_e^{\otimes n}(\cdot\vert\x^{l_0l_1},\y^{l_0})1_{E\2_1(\u^{l_0},\x^{l_0l_1},\y^{l_0})}(1-1_{F\2_1(\u^{l_0},\y^{l_0})})\rVert\label{2III}\\
  &+\lVert\frac{1}{L_0L_1}\sum_{l_0,l_1}W_e^{\otimes n}(\cdot\vert\x^{l_0l_1},\y^{l_0})(1-1_{E\2_1(\u^{l_0},\x^{l_0l_1},\y^{l_0})})\rVert.\label{2IV}
\end{align}
Due to \eqref{cond2,2}, we know that $\eqref{2I}\leq\eps$.

Next we consider \eqref{2III}. Due to \eqref{cond2,1}, we have
\begin{align*}
  \eqref{2III}&\leq1-\frac{1}{L_0L_1}\sum_{l_0,l_1}W_e^{\otimes n}(E\2_0(\u^{l_0},\x^{l_0l_1},\y^{l_0})\vert\x^{l_0l_1},\y^{l_0})\\
  &\leq1-\frac{1-\eps}{L_0}\sum_{l_0}\hat\vartheta\2_{\u^{l_0}\y^{l_0}}(\Z^n).
\end{align*}
As for Case 3, one lower-bounds $\hat\vartheta\2_{\u^{l_0}\y^{l_0}}(\Z^n)\geq\vartheta\2_{\u^{l_0}\y^{l_0}}(\Z^n)-\eps$ by $1-2\cdot2^{-nc\delta^2}-\eps$. Thus we can conclude that
\[
  \eqref{2III}\leq2(\eps+2^{-nc\delta^2}).
\]

For \eqref{2II}, we have by \eqref{cond2,1}
\begin{align*}
  \eqref{2II}&\leq1-\frac{1}{L_0L_1}\sum_{l_0,l_1}W_e^{\otimes n}(E\2_0(\u^{l_0},\x^{l_0l_1},\y^{l_0})\cap F\2_0\vert\x^{l_0l_1},\y^{l_0})\\
  &\leq1-(1-3\eps)\hat\vartheta\2(F\2_0).
\end{align*}
It remains to lower-bound $\hat\vartheta\2(F\2_0)\geq\vartheta\2(\Z^n)-\eps$. As in the lower bound on $\theta\1(\Z^n)$ above, one obtains the bound
\[
  \vartheta\2(\Z^n)\geq1-2\cdot2^{-nc\delta^2}-\eps.
\]
Thus we conclude
\[
  \eqref{2II}\leq5\eps+2\cdot2^{-nc\delta^2}.
\]

Finally, we use condition \eqref{cond2,*} to bound \eqref{2IV}. We have
\begin{align}
  \eqref{2IV}&=\frac{1}{L_0L_1}\sum_{l_0,l_1}W_e^{\otimes n}(E\2_1(\u^{l_0},\x^{l_0l_1},\y^{l_0})\vert \x^{l_0l_1},\y^{l_0})\notag\\
  =&\,\frac{1}{L_0}\sum_{l_0}\biggl(\frac{1}{L_1}\sum_{l_1:\x^{l_0l_1}\in T_{X\vert YU,\delta}^n(\y^{l_0},\u^{l_0})}W_e^{\otimes n}(E\2_1(\u^{l_0},\x^{l_0l_1},\y^{l_0})\vert \x^{l_0l_1},\y^{l_0})\label{2l_0sum1}\\
  &+\frac{1}{L_1}\sum_{l_1:\x^{l_0l_1}\notin T_{X\vert YU,\delta}^n(\y^{l_0},\u^{l_0})}W_e^{\otimes n}(E\2_1(\u^{l_0},\x^{l_0l_1},\y^{l_0})\vert \x^{l_0l_1},\y^{l_0})\biggr).\label{2l_0sum2}
\end{align}
For every $l_0$, the summand appearing in \eqref{2l_0sum1} can be upper-bounded by $2^{-nc\delta^2}$. By assumption \eqref{cond2,*}, \eqref{2l_0sum2} is upper-bounded by $1-(1-\eps)(1-2\cdot2^{-nc\delta^2})$. Thus
\[
  \eqref{2IV}\leq\eps+3\cdot2^{-nc\delta^2}.
\]

Collecting the bounds for \eqref{2I}-\eqref{2IV}, we obtain a total upper bound of
\[
  9\eps+7\cdot2^{-nc\delta^2}+L_0\lvert\Z\rvert^np\2.
\]
This finishes the proof of Lemma \ref{sec2}.

\subsubsection{Bounds for Case 3:}

Now we treat the case $L_1=L_2=1$. Lemma \ref{lemgemtyp1} is the analog of Lemma \ref{lemgemtyp3}, the proofs are analogous.

\begin{lemma}\label{lemgemtyp1}
  Let the event $A\3_*$ be defined by
\begin{multline*}\label{gemtyp}
  A\3_*:=\bigl\{\lvert\{l_0\in[L_0]:X^{l_0}\in T_{X\vert YU,\delta}^{n}(Y^{l_0},U^{l_0})\}\rvert\geq(1-\eps)(1-2\cdot2^{-{n}c_1\delta^2})L_0\}.
\end{multline*}
Then
\[
  \PP[(A\3_*)^c]\leq\exp\left(-L_0\cdot\frac{\eps^2(1-2\cdot2^{-{n}c_1\delta^2})}{2\ln 2}\right).
\]
\end{lemma}

Let 
\[
  E\3(\x,\y):=\{\z\in T_{Z,4\lvert\Y\rvert\lvert\X\rvert\lvert\U\rvert\delta}^{n}:W_e^{\otimes {n}}(\z\vert \x,\y)\leq2^{-{n}(H(Z\vert XY)-f_2(\delta))}\},
\]
where $f_2(\delta)=\tau(P_{XYZ},3\delta,\delta)$. Let 
\[
  \vartheta\3(\z):=\EE[W_e^{\otimes n}(\z\vert X^{1},Y^{1})1_{E_1(X^{1},Y^{1})}(\z)]
\]
and for
\[
  F\3:=\{\z\in T_{Z,4\lvert\Y\rvert\lvert\X\rvert\lvert\U\rvert\delta}^{n}:\vartheta(\z)\geq \eps\lvert T_{Z,\delta}^{n}\rvert^{-1}\}
\]
define the measure
\[
  \hat\vartheta\3:=\hat\vartheta\3\cdot1_{F\3}.
\]

\begin{lemma}\label{AW1}
  Let $\z\in F\3$. Let $A\3(\z)$ be the event that
\[
  \frac{1}{L_0}\sum_{l_0}W_e^{\otimes {n}}(\z\vert X^{l_0},Y^{l_0})1_{E\3(X^{l_0},Y^{l_0})}(\z)\in[(1\pm\eps)\hat\vartheta\3(\z)].
\]
Then for $f_1(\delta)=\tau(P_{UYZ},4\lvert\Y\rvert\lvert\X\rvert\lvert\U\rvert\delta,\delta)$,
\begin{align*}
  \PP[A\3(\z)^c]\leq 2\exp\left(-L_0\cdot\frac{\eps^32^{-{n}(I(Z\wedge XY)+f_1(\delta)+f_2(\delta))}}{2\ln 2}\right).
\end{align*}
\end{lemma}

The proof of this lemma is analogous to that of Lemma \ref{AW3}.

\subsubsection{Secrecy for Case 3:}\label{subsect:sec}

\begin{lemma}\label{sec1}
Let $\{(\u^{l_0},\x^{l_0},\y^{l_0})\}$ be a realization of $\F$ satisfying the conditions of 
\begin{align}
  &A\3_*,\label{cond1,*}\\
  \bigcap_{\z\in F\3}&A\3(\z).\label{cond1,1}
\end{align}
Then for sufficiently large $n$,
\begin{align}
  \lVert\hat\vartheta\3-\frac{1}{L_0}\sum_{l_0}W_e^{\otimes n}(\cdot\vert\x^{l_0},\y^{l_0})\rVert\leq4\eps+5\cdot2^{-nc\delta^2}.
\end{align}
\end{lemma}

We use this subsection to prove the above lemma. We have
\begin{align}
  &\quad\;\lVert\hat\vartheta\3-\frac{1}{L_0}\sum_{l_0}W_e^{\otimes n}(\cdot\vert\x^{l_0},\y^{l_0})\rVert\notag\\
  &\leq\lVert\hat\vartheta\3-\frac{1}{L_0}\sum_{l_0}W_e^{\otimes n}(\cdot\vert\x^{l_0},\y^{l_0})1_{E\3(\x^{l_0},\y^{l_0})}1_{F\3}\rVert\label{1I}\\
  &+\lVert\frac{1}{L_0}\sum_{l_0}W_e^{\otimes n}(\cdot\vert\x^{l_0},\y^{l_0})1_{E\3(\x^{l_0},\y^{l_0})}(1-1_{F\3})\rVert\label{1II}\\
  &+\lVert\frac{1}{L_0}\sum_{l_0}W_e^{\otimes n}(\cdot\vert\x^{l_0},\y^{l_0})(1-1_{E\3(\x^{l_0},\y^{l_0})})\rVert.\label{1III}
\end{align}
Due to \eqref{cond1,1} we have $\eqref{1I}\leq\eps$.

Next we bound \eqref{1II}. Again using \eqref{cond1,1},
\begin{align}
  \eqref{1II}&\leq 1-\frac{1}{L_0}\sum_{l_0}W_e^{\otimes n}(E\3(\x^{l_0},\y^{l_0})\cap F\3\vert\x^{l_0},\y^{l_0})\notag\\
  &\leq1-(1-\eps)\hat\vartheta\3(F\3).\label{theta1F}
\end{align}
As in Case 1 and 2, $\hat\vartheta\3(F\3)$ can be lower-bounded by $1-2\cdot2^{-nc\delta^2}-\eps$, so
\[
  \eqref{1II}\leq1-(1-\eps)(1-2\cdot2^{-nc\delta^2}-\eps)\leq2(\eps+2^{-nc\delta^2}).
\]

Finally, the third term \eqref{1III} equals
\begin{align}
  &\quad\;\frac{1}{L_0}\sum_{l_0}W_e^{\otimes n}(E\3(\x^{l_0},\y^{l_0})^c\vert\x^{l_0},\y^{l_0})\notag\\
  &=\frac{1}{L_0}\sum_{l_0:\x^{l_0}\in T_{X\vert YU,\delta}(\y^{l_0},\u^{l_0})}W_e^{\otimes n}(E\3(\x^{l_0},\y^{l_0})^c\vert\x^{l_0},\y^{l_0})\label{11rhs}\\
  &+\frac{1}{L_0}\sum_{l_0:\x^{l_0}\notin T_{X\vert YU,\delta}(\y^{l_0},\u^{l_0})}W_e^{\otimes n}(E\3(\x^{l_0},\y^{l_0})^c\vert\x^{l_0},\y^{l_0}).\label{12rhs}
\end{align}
and is lower-bounded by
\[
  \eqref{1III}\leq 2^{-nc\delta^2}+(1-\eps)(1-2\cdot2^{-nc\delta^2})\leq\eps+3\cdot2^{-nc\delta^2}.
\]

Combining the above bounds, we can conclude that
\[
  \eqref{1I}+\eqref{1II}+\eqref{1III}\leq4\eps+5\cdot2^{-nc\delta^2},
\]
which completes the proof of Lemma \ref{sec1}.

\subsection{Random Coding for the Non-Wiretap MAC with Common Message}\label{subsect:randcodcommmess}

Assume we are given another family of random variables 
\[
  \F':=\bigcup_{l_0\in[L_0]}(U^{l_0'},\F'_{l_0'})
\]
with $\F'_{l_0'}=\{X^{l_0'l_1'},Y^{l_0'l_2'}:l_1',l_2'\in[L_1']\times[L_2']\}$ for other positive integers $L_0',L_1',L_2'$ with blocklength $n'$ which is  independent of $\F$, but which has the same structure as $\F$ and whose distribution is defined according to the same $p$ as $\F$. Define the rate set $\tilde\R(p)$ by the bounds
\begin{eqnarray*}
  \tilde R_1\leq& I(T\wedge X\vert YU),\\
  \tilde R_2\leq& I(T\wedge Y\vert XU),\\
  \tilde R_1+\tilde R_2\leq& I(T\wedge XY\vert U),\\
  \tilde R_0+\tilde R_1+\tilde R_2\leq& I(T\wedge XY).
\end{eqnarray*}
Assume that for some $0<\eta<I_*:=\min\{I_\nu>0:\nu=1,2,3,4\}$ we have
\begin{align*}
  \frac{n\log L_1+n'\log L_1'}{n+n'}&\leq [I(T\wedge X\vert YU)-\eta\,]_+,\\
  \frac{n\log L_2+n'\log L_2'}{n+n'}&\leq [I(T\wedge Y\vert XU)-\eta\,]_+,\\
  \frac{n\log(L_1L_2)+n'\log(L_1'L_2')}{n+n'}&\leq [I(T\wedge XY\vert U)-\eta\,]_+,\\
  \frac{n\log(L_0L_1L_2)+n'\log(L_0'L_1'L_2')}{n+n'}&\leq [I(T\wedge XY)-\eta\,]_+.
\end{align*}
Define a new family of random vectors
\[
  \F\circ\F':=\{\tilde U^{l_0l_0'},\tilde X^{l_0l_0'l_1l_1'},\tilde Y^{l_0l_0'l_2l_2'}\}
\]
by concatenating the corresponding elements of $\F$ and $\F'$, so e.g. $\tilde U^{l_0l_0'}=(U^{l_0},U^{l_0'})\in\U^{n+n'}$, $\tilde X^{l_0l_0'l_1l_1'}=(X^{l_0l_1},X^{l_0'l_1'})\in\X^{n+n'}$.

\begin{lemma}\label{MACcode}
  For any $\delta,\eta>0$ there are $\zeta_1,\zeta_2=\zeta_1(\eta,\delta),\zeta_2(\eta,\delta)>0$ such that the probability of the event $A_{\MAC}$ that the family 
\[
  \{\tilde X^{l_0l_0'l_1l_1'},\tilde Y^{l_0l_0'l_2l_2'}:(l_0,l_0',l_1,l_1',l_2,l_2')\}
\]
is the codeword set of a deterministic MAC code with average error at most $\exp(-(n+n')\zeta_1)$ is lower-bounded by  $1-\exp(-(n+n')\zeta_2)$. The same result is true if it is formulated only for $\F$ or $\F'$ without concatenation.
\end{lemma}
\begin{proof}

The difference to standard random coding proofs is that the random variables from $\F$ and $\F'$ are conditioned on typicality. Using the random sets
\[
  E^{l_0l_0'l_1l_1'l_2l_2'}:=
  \{\t\in\T^n:(\tilde U^{l_0l_0'},\tilde X^{l_0l_0'l_1l_1'},\tilde Y^{l_0l_0'l_2l_2'},\t)\in T_{UXYT,\delta}^n\},
\]
we define the decoding sets $F^{l_0l_0'l_1l_1'l_2l_2'}$ by deciding for $(l_0,l_0',l_1,l_1',l_2,l_2')$ if the output is contained in $E^{l_0l_0'l_1l_1'l_2l_2'}$ and if at the same time it is not contained in any $E^{\tilde l_0\tilde l_0'\tilde l_1\tilde l_1'\tilde l_2\tilde l_2'}$ for a different message tuple $(\tilde l_0,\tilde l_0',\tilde l_1,\tilde l_1',\tilde l_2,\tilde l_2')$. This decoder is known to be the right decoder in the case where the codewords have the standard i.i.d.\ structure, i.e. for a family of random variables
\[
  \{\hat U^{l_0l_0'},\hat X^{l_0l_0'l_1l_1'},Y^{l_0l_0'l_2l_2'}\}
\]
where $\hat U^{l_0l_0'}\sim P_{U}^{\otimes(n+n')}$ and where conditional on $\hat U^{l_0l_0'}$, the $\hat X^{l_0l_0'l_1l_1'}$ and $\hat Y^{l_0l_0'l_2l_2'}$ are independent with $\hat X^{l_0l_0'l_1l_1'}\sim P_{X\vert U}^{\otimes (n+n')}$ and $\hat Y^{l_0l_0'l_2l_2'}\sim P_{Y\vert U}^{\otimes (n+n')}$. It is easily seen that 
\begin{multline*}
  \EE[W^{\otimes n}((F^{l_0l_0'l_1l_1'l_2l_2'})^c\vert\tilde X^{l_0l_0'l_1l_1'},\tilde Y^{l_0l_0'l_2l_2'})]\\
  \leq(1-2^{-nc\delta^2})^3(1-2^{-n'c\delta^2})^3\EE[W^{\otimes n}((F^{l_0l_0'l_1l_1'l_2l_2'})^c\vert\hat X^{l_0l_0'l_1l_1'},\hat Y^{l_0l_0'l_2l_2'})].
\end{multline*}
Then the standard random coding proof technique yields the result. The specialization for the case that only $\F$ or $\F'$ is treated is obvious.\qed
\end{proof}

\subsection{Coding}\label{MAC}

In this subsection we show the achievability of the rate sets $\R^{(\nu)}(p)$ for $\nu=0,1,2,3$ and appropriate $p$. For the cases where we showed that $\R^{(\nu)}(p)$ can be written as the union over certain $\alpha$ of rate sets $\R^{(\nu)}_\alpha(p)$, we show the achievability of the latter for every $\alpha$. 

Throughout this section fix a common randomness bound $H_C\geq 0$. Let $\delta>0$ which will be specified later and $n$ a blocklength which will have to be large enough. Every $p$ considered in this section has the form $p=P_U\otimes(P_{X\vert U}\otimes P_{Y\vert U})\otimes W$. Without loss of generality we may assume that $I(Z\wedge XY)<I(T\wedge XY)$, in particular, $I(T\wedge XY)>0$. Letting
\begin{equation}\label{parameters}
  K_0,K_1,K_2,L_0,L_1,L_2,n,\qquad K_0',K_1',K_2',L_0',L_1',L_2',n'
\end{equation}
be arbitrary positive integers, we define two independent families $\G,\G'$ of random vectors. $\G$ has the same form as $\F$ with the parameters $L_0,L_1,L_2$ replaced by $K_0L_0,K_1L_1,K_2L_2$. $\G'$ is defined analogously with the parameters on the left-hand side of \eqref{parameters} replaced by those on its right-hand side. Every choice of $(k_0,k_1,k_2)$ induces a subfamily $\F$ of $\G$ which has the same parameters as the $\F$ treated above, every subfamily of $\G'$ corresponding to any $(k_0',k_1',k_2')$ induces an $\F'$ with parameters $L_0',L_1',L_2',n'$. Further recall the notation $\G\circ\G'$ as the family of concatenated words from $\G$ and $\G'$.

\subsubsection{Case 0 and 1:}

Let $p\in\Pi_0$ or $p\in\Pi\1_{H_C}$. Note that $\alpha\1_0\leq\alpha\1_1$ if and only if the vector $(J_0\halpha,J_1\halpha,J_2\halpha)$ whose components are given by
\begin{align*}
  J_0\halpha&=I(Z\wedge U),\\
  J_1\halpha&=\alpha I(Z\wedge X\vert YU)+(1-\alpha)I(Z\wedge X\vert U),\\
  J_2\halpha&=\alpha I(Z\wedge Y\vert U)+(1-\alpha)I(Z\wedge Y\vert XU)
\end{align*}
is contained in $\tilde\R(p)$. We first consider Case 1. Let a rate vector $(R_0,R_1,R_2)$ with positive components be given such that $(\tilde R_0,\tilde R_1,\tilde R_2):=(R_0,R_1,R_2)+(J_0\halpha,J_1\halpha,J_2\halpha)\in\tilde\R(p)$, which means that $(R_0,R_1,R_2)\in\R_\alpha(p)$. We now define a wiretap code whose rates approximate $(R_0,R_1,R_2)$. If $\alpha=0$, we only need $\G'$, if $\alpha=1$, we only need $\G$. Otherwise we do time-sharing in the following way: choose for a small $0<\gamma<\min\{\alpha,1-\alpha\}$ blocklengths $n$ and $n'$ with $n/(n+n')\in(\alpha-\gamma,\alpha+\gamma)$. For some $0<2\eta<\min\{R_0,R_1,R_2\}$ and every $\nu=0,1,2$ let
\[
  \tilde R_\nu-\eta\leq\frac{\log(K_\nu L_\nu)+\log(K_\nu'L_\nu')}{n+n'}\leq\tilde R_\nu-\frac{\eta}{2}
\]
(and this modifies accordingly for $\alpha\in\{0,1\}$). By Lemma \ref{MACcode} we know that with probability exponentially close to 1, the random variables $\tilde X_{k_0k_0'k_1k_1'}^{l_0l_0'l_1l_1'}$ and $\tilde Y_{k_0k_0'k_2k_2'}^{l_0l_0'l_2l_2'}$ form the codewords of a code for the non-wiretap MAC given by $W_b$ with an average error at most $\exp(-(n+n')\zeta_1)$ for some $\zeta_1>0$. We denote Bob's corresponding random decoder by $\Phi$. Now  let
\begin{align*}
  \frac{\log L_1+\log L_1'}{n+n'}&\in J_1\halpha+\bigl(f_1(\delta)+(\alpha f_2(\delta)+(1-\alpha)f_4(\delta))\bigr)\cdot[2,3],\\
  \frac{\log L_2+\log L_2'}{n+n'}&\in J_2\halpha+\bigl(f_1(\delta)+(\alpha f_4(\delta)+(1-\alpha)f_2(\delta))\bigr)\cdot[2,3],\\
  \frac{\log L_0+\log L_0'}{n+n'}&\in J_0\halpha+\bigl(f_4(\delta)+f_6(\delta)\bigr)\cdot[2,3].
\end{align*}
This is possible if $4(f_1(\delta)+f_2(\delta)+f_4(\delta))\leq\min\{\eta,H_C-J_0\halpha\}$. If additionally $\eps$ is chosen according to
\[
  -\frac{1}{n}\log\eps=\frac{1}{4}\min\{4\zeta_1,f_1(\delta)+f_2(\delta)+f_4(\delta)+f_6(\delta)\},
\]
then for every $(k_0,k_1,k_2)\in[K_0]\times[K_1]\times[K_2]$, the corresponding subfamily $\F$ of $\G$ satisfies \eqref{cond3,*}-\eqref{cond3,3} with probability exponentially close to 1, and for every $(k_0',k_1',k_2')\in[K_0']\times[K_1']\times[K_2']$, the corresponding subfamily $\F'$ of $\G'$ satisfies (\ref{cond3,*}$'$)-(\ref{cond3,3}$'$) with probability exponentially close to 1. Thus we can choose a realization of $\G\circ\G'$ which has all these properties and use it to define a stochastic wiretap code. We define independent encoders $G$ and $G'$ by setting
\begin{align*}
  G_0(l_0\vert k_0)&=\frac{1}{L_0},&(k_0\in[K_0],l_0\in[L_0])&,\\
  G_1(\x\vert k_0,k_1,l_0)&=\frac{1}{L_1}\sum_{l_1}\delta_{\x_{k_0k_1}^{l_0l_1}}(\x),&(\x\in\X^n,k_1\in[K_1],k_0\in[K_0],l_0\in[L_0])&,\\
  G_2(\y\vert k_0,k_2,l_0)&=\frac{1}{L_2}\sum_{l_2}\delta_{\y_{k_0k_2}^{l_0l_2}}(\y),&(\y\in\Y^n,k_2\in[K_2],k_0\in[K_0],l_0\in[L_0])&,
\end{align*}
and defining $G'$ analogously. $G_0$ and $G_0'$ satisfy the common randomness constraint. We choose the decoder $\phi$ to be the realization of $\Phi$ corresponding to the chosen realization of $\G\circ\G'$. The average error of the stochastic encoding code equals the average error of the deterministic MAC code for $W_b$ determined by the realization of $\G\circ\G'$, in particular it is bounded by $\eps$. Due to the choice of $\delta$ the rates of this code satisfy
\begin{align*}
  \frac{\log K_\nu+\log K_\nu'}{n+n'}\geq R_\nu-2\eta&&(\nu=0,1,2,).
\end{align*}
Finally if we let $M_\nu$ be uniformly distributed on $[K_\nu]$ and $M_\nu'$ on $[K_\nu']$, then it follows from Lemma \ref{sec3} and \eqref{totvarsmall} together with the fact that $\eps$ is exponentially small that the strong secrecy criterion is satisfied. Thus the rate triple $(R_0,R_1,R_2)$ is achievable. So far, this excludes $(R_0,R_1,R_2)$ where one component equals zero, but as $\delta$ and $\eta$ may be arbitrarily close to 0 and the achievable region of $W$ is closed by definition, we can conclude that the whole region $\R_\alpha(p)$ is achievable.

For Case 0, everything goes through if one sets $K_0=K_0'=L_0=L_0'=1$ and $R_0=0$. The difference to Case 1 is that even if $J_0\halpha=0$, one needs a little bit more common randomness than that in order to protect a common message, as can be seen in the choice of $L_0$ and $L_0'$ above. Thus the transmission of a common message is impossible if common randomness is not available.

\subsubsection{Case 2:}

Let $p\in\Pi_{H_C}\2$. In this case we generally need both a $\G$ and a $\G'$, where $\G$ has $L_2=1$ and $\G'$ has $L_1=1$. We define the vector $(J_0\halpha,J_1\halpha,J_2\halpha)$ by 
\begin{align*}
  J_0\halpha&=\alpha I(Z\wedge YU)+(1-\alpha)I(Z\wedge XU),\\
  J_1\halpha&=\alpha I(Z\wedge X\vert YU),\\
  J_2\halpha&=(1-\alpha)I(Z\wedge Y\vert XU)
\end{align*}
As it should always be clear which case we are treating, this should not lead to confusion with case 1. Note that $\alpha\2_0\leq\alpha\leq\alpha\2_1$ if and only if $(J_0\halpha,J_1\halpha,J_2\halpha)$ is contained in $\tilde\R(p)$ and satisfies $J_0\halpha<H_C$. Let a rate vector $(R_0,R_1,R_2)$ be given whose $\nu$-th component may only vanish if $L_\nu=L_\nu'=1$. Further we require that $(\tilde R_0,\tilde R_1,\tilde R_2)=(R_0,R_1,R_2)+(J_0\halpha,J_1\halpha,J_2\halpha)$ is contained in $\tilde\R(p)$. If $\alpha=0$, we only need $\G'$, if $\alpha=1$, we only need $\G$. Otherwise, let $0<\gamma<\min\{\alpha,1-\alpha\}$ be small and let $n$ and $n'$ be large enough such that $n/(n+n')\in(\alpha-\gamma,\alpha+\gamma)$. Further for some $0<2\eta<\min\{R_\nu:\nu=0,1,2,R_\nu>0\}$ let
\begin{align*}
  [\tilde R_\nu-\eta]_+\leq\frac{\log (K_\nu L_\nu)+\log (K_\nu'L_\nu')}{n+n'}\leq[\tilde R_\nu-\frac{\eta}{2}]_+,
\end{align*}
and modify this accordingly for $\alpha\in\{0,1\}$. By Lemma \ref{MACcode} we know that with probability exponentially close to 1, the random variables $\tilde X_{k_0k_0'k_1k_1'}^{l_0l_0'l_1l_1'}$ and $\tilde Y_{k_0k_0'k_2k_2'}^{l_0l_0'l_2l_2'}$ form the codewords of a code for the non-wiretap MAC given by $W_b$ with an average error at most $\exp(-(n+n')\zeta_1)$ for some $\zeta_1>0$. We denote the corresponding random decoder by $\Phi$. We define $(j_1^1,j_1^2)=(j_2^1,j_2^2)=(1,2)$ and $(j_0^1,j_0^2)=(1,6)$. Then let for $\nu=0,1,2$ 
\begin{align*}
  J_\nu\halpha+2(f_{j_\nu^1}(\delta)+f_{j_\nu^2}(\delta))\leq\frac{\log L_\nu+\log L_\nu'}{n+n'}\leq J_\nu\halpha+3(f_{j_\nu^1}(\delta)+f_{j_\nu^2}(\delta)),
\end{align*}
which is possible if $4(f_{j_\nu^1}(\delta)+f_{j_\nu^2}(\delta))\leq\min\{\eta,H_C-J_0\halpha\}$ for all $\nu$. If additionally $\eps$ is chosen according to
\[
  -\frac{1}{n}\log\eps=\frac{1}{4}\min\{4\zeta_1,f_1(\delta)+f_2(\delta),f_1(\delta)+f_6(\delta)\},
\]
then for every $(k_0,k_1,k_2)\in[K_0]\times[K_1]\times[K_2]$, the corresponding subfamily $\F$ of $\G$ satisfies \eqref{cond2,*}-\eqref{cond2,2} with probability exponentially close to 1, and for every $(k_0',k_1',k_2')\in[K_0']\times[K_1']\times[K_2']$, the corresponding subfamily $\F'$ of $\G'$ satisfies (\ref{cond2,*}$'$)-(\ref{cond2,2}$'$) with probability exponentially close to 1. Thus we can choose a realization of $\G\circ\G'$ which has all these properties plus those defining $A_\MAC$ and use it to define a stochastic wiretap code. We define independent encoders $G$ and $G'$ by setting
\begin{align*}
  G_0(l_0\vert k_0)&=\frac{1}{L_0},&(l_0\in[L_0],k_0\in[K_0])&,\\
  G_1(\x\vert k_0,k_1,l_0)&=\frac{1}{L_1}\sum_{l_1}\delta_{\x_{k_0k_1}^{l_0l_1}}(\x),&(\x\in\X^n,k_1\in[K_1],k_0\in[K_0],l_0\in[L_0])&,\\
  G_2(\y\vert k_0,k_2,l_0)&=\delta_{\y_{k_0k_2}^{l_0}}(\y),&(\y\in\Y^n,k_2\in[K_2],k_0\in[K_0],l_0\in[L_0])&,
\end{align*}
and defining $G'$ analogously. The decoder $\phi$ is the realization of $\Phi$ corresponding to the chosen realization of $\G\circ\G'$. $G_0$ and $G_0'$ satisfy the common randomness constraint. Due to the simple form of $G$ and $G'$, the average error of the stochastic encoding code equals the average error of the deterministic MAC code for $W_b$ determined by the realization of $\G\circ\G'$, in particular it is bounded by $\eps$. Due to the choice of $\delta$, the rates of this code satisfy
\begin{align*}
  \frac{\log K_\nu+\log K_\nu'}{n+n'}\geq R_\nu-2\eta&&(\nu=0,1,2,).
\end{align*}
Finally if we let $M_\nu$ be uniformly distributed on $[K_\nu]$ and $M_\nu'$ on $[K_\nu']$, then it follows from Lemma \ref{sec3} and \eqref{totvarsmall} together with the fact that $\eps$ is exponentially small that the strong secrecy criterion is satisfied. Thus the rate triple $(R_0,R_1,R_2)$ is achievable. So far, this may exclude rate triples $(R_0,R_1,R_2)$ where one component equals zero, but as $\delta$ and $\eta$ may be arbitrarily close to 0 and the achievable region of $W$ is closed by definition, we can conclude that the whole region $\R_\alpha(p)$ is achievable.

\subsubsection{Case 3:}

In this case we only need $\G$ with $L_1=L_2=1$. Let $R_0>0$ and assume that the rate vector $(\tilde R_0,\tilde R_1,\tilde R_2):=(R_0+I(Z\wedge XY),R_1,R_2)$ is contained in $\tilde\R(p)$. Further for some  $0<2\eta<\min\{R_\nu:\nu=0,1,2,R_\nu>0\}$ let
\begin{align*}
  [\tilde R_\nu-\eta]_+\leq\frac{1}{n}\log (K_\nu L_\nu)\leq[\tilde R_\nu-\frac{\eta}{2}]_+.
\end{align*}
$\G$ satisfies $A_{\MAC}$ with probability exponentially close to 1, so the $X_{k_0k_1}^{l_0l_1}$ and $Y_{k_0k_2}^{l_0l_2}$ form the codewords of a deterministic non-wiretap MAC code whose average error for transmission over $W_b$ is bounded by $\exp(-n\zeta_1)$ for some $\zeta_1>0$. We denote the corresponding random decoder by $\Phi$. Now let 
\[
  I(Z\wedge XY)+2(f_1(\delta)+f_2(\delta))\leq\frac{1}{n}\log L_0\leq I(Z\wedge XY)+3(f_1(\delta)+f_2(\delta))
\]
for $\delta$ so small that $4(f_1(\delta)+f_2(\delta))\leq\min(\eta,H_C-I(Z\wedge XY))$ and choose $\eps$ such that
\[
  -\frac{1}{n}\log\eps=\frac{1}{4}\min\{4\zeta_1,f_1(\delta)+f_2(\delta)\}.
\]
Then for every $(k_0,k_1,k_2)$ the corresponding family $\F$ satisfies the conditions \eqref{cond1,*} and \eqref{cond1,1} with probability exponentially close to 1. We choose a realization $\{(\u_{k_0}^{l_0},\x_{k_0k_1}^{l_0},\y_{k_0k_2}^{l_0})\}$ which satisfies the conditions of \eqref{cond1,*} and \eqref{cond1,1} and which determines a deterministic non-wiretap code for $W_b$ with decoder $\phi$. Now we can define a wiretap code whose decoder is $\phi$ and whose stochastic encoder $G$ is given by
\begin{align*}
  G_0(l_0\vert k_0)&=\frac{1}{L_0}, & (k_0\in[K_0],l_0\in[L_0])&,\\
  G_1(\x\vert k_0,k_1,l_0)&=\delta_{\x_{k_0k_1}^{l_0}}(\x),& (\x\in\X^n,k_1\in[K_1],k_0\in[K_0],l_0\in[L_0])&,\\
  G_2(\y\vert k_0,k_2,l_0)&=\delta_{\y_{k_0k_2}^{l_0}}(\y),& (\y\in\Y^n,k_2\in[K_2],k_0\in[K_0],l_0\in[L_0])&.
\end{align*}
Note that $G_0$ satisfies the common randomness constraint. Due to the uniform distribution of $G_0$, its average error is identical to that of the deterministic MAC code determined by the $\x_{k_0k_1}^{l_0}$ and the $\y_{k_0k_2}^{l_0}$, in particular, it is exponentially small with rate at most $\eps$. We have for $\nu=0,1,2$
\[
  \frac{1}{n}\log K_\nu\geq R_\nu-2\eta.
\]
due to the choice of $\delta$. Finally if we let $M_\nu$ be uniformly distributed on $[K_\nu]$, then it follows from Lemma \ref{sec1} and \eqref{totvarsmall} together with the fact that $\eps$ is exponentially small that the strong secrecy criterion is satisfied. Thus the rate triple $(R_0,R_1,R_2)$, and hence $\R(p)$, is achievable.

\subsection{Concluding Steps}\label{subsect:concl}

We can reduce coding for a general $p$ which is the distribution of a random vector $(U,V_1,V_2,X,Y,T,Z)$ to the case treated above by constructing a new wiretap MAC as follows: its input alphabets are $\V_1$ and $\V_2$, its output alphabets still are $\T$ and $\Z$. The transition probability for inputs $(v_1,v_2)$ and outputs $(t,z)$ is given by
\[
  \tilde W(t,z\vert v_1,v_2):=\sum_{(x,y)\in\X\times\Y}W(t,z\vert x,y)P_{X\vert V_1}(x\vert v_1)P_{Y\vert V_2}(y\vert v_2).
\]
For this channel we do the same construction as above considering the joint distribution of random variables $(U,V_1,V_2,T,Z)$ which we denote by $\tilde p$. In this way we also construct a wiretap code for the original channel $W$ because the additional randomness $P_{V_1V_2\vert U}$ can be integrated into the stochastic encoders $G_1$ and $G_2$. $G_0$ remains unchanged, so the additional randomness in the encoders does not increase the common randomness needed to do the encoding.

On the other hand, we need to show that the rate regions thus obtained are those appearing in the statement of Theorem \ref{thmcomm}. As the sets $\Pi_0,\Pi_{H_C}\1,\ldots,$ $\Pi_{H_C}\3$ depend on the channel, we write $\Pi_0(W)$, $\Pi_0(\tilde W)$, $\Pi_{H_C}\1(W),\ldots,\Pi_{H_C}\3(\tilde W)$. Note that $\tilde p$ is contained in $\Pi_0(\tilde W)$ or $\Pi_{H_C}^{(\nu)}(\tilde W)$ for some $\nu=1,2,3$ if and only if $p$ is contained in the corresponding $\Pi_0(W)$ or $\Pi_{H_C}^{(\nu)}(W)$. This immediately implies that the rate regions also coincide.

\section{Proof of Theorem \ref{thmconf}}\label{sect:proofconf}

\subsection{Elementary Rate Regions}\label{subsect:elratreg}

As for the wiretap MAC with common message we show that we can write the claimed achievable regions as unions of simpler sets whose achievability will be show in the next step.
\subsubsection{For Case 1:}

Define
\[
  \beta\1_0:=[1-\frac{C_2}{I(Z\wedge U)}]_+,\qquad\beta\1_1:=\min\{\frac{C_1}{I(Z\wedge U)},1\}.
\]
We have $\beta\1_0\leq\beta\1_1$ because $I(Z\wedge U)<C_1+C_2$.

\begin{lemma}\label{alfaromeo}
  For $\beta\1_0\leq\beta\leq\beta\1_1$, let $\R\1_\beta(p,C_1,C_2)$ be the set of those real pairs $(R_1,R_2)$ satisfying
\begin{align*}
  R_1&\leq I(T\wedge V_1\vert V_2U)-I(Z\wedge V_1\vert U)\\&\hspace{.7cm}-[I(Z\wedge V_2\vert V_1U)-I(T\wedge V_2\vert V_1U)]_+-\beta I(Z\wedge U)+C_1,\\
  R_2&\leq I(T\wedge V_2\vert V_1U)-I(Z\wedge V_2\vert U)\\&\hspace{.7cm}-[I(Z\wedge V_1\vert V_2U)-I(T\wedge V_1\vert V_2U)]_+-(1-\beta) I(Z\wedge U)+C_2,\\
  R_1+R_2&\leq\min\bigl\{I(T\wedge V_1V_2\vert U)-I(Z\wedge V_1V_2\vert U)-I(Z\wedge U)+C_1+C_2,\\
  &\qquad\qquad I(T\wedge V_1V_2)-I(Z\wedge V_1V_2)\bigr\}.
\end{align*}
Then
\[
  \R\1(p,C_1,C_2)=\bigcup_{\beta\1_0\leq\beta\leq\beta\1_1}\R\1_\beta(p,C_1,C_2).
\]
\end{lemma}

Thus it is sufficient to show the achievability of $\R\1_\beta(p,C_1,C_2)$ for every $\beta$. For the proof one uses Lemma \ref{gemconc}.

\subsubsection{For Case 2:}

Recall the vector $(J_0\halpha,J_1\halpha,J_2\halpha)$ defined as in Case 2 from the common message part. Define
\[
  \beta\zhalpha_0:=[1-\frac{C_2}{J_0\halpha}]_+,\qquad\beta\zhalpha_1:=\min\{\frac{C_1}{J_0\halpha},1\}.
\]
We show that every $\R\2_\alpha(p,C_1,C_2)$ can be represented as the union of sets $\R\2_{\alpha,\beta}(p,C_1,C_2)$ for $\beta\zhalpha_0\leq\beta\leq\beta\zhalpha_1$. Define $\R\2_{\alpha,\beta}(p,C_1,C_2)$ by
\begin{align*}
  R_1&\leq I(T\wedge V_1\vert V_2U)-\alpha I(Z\wedge V_1\vert V_2U)+C_1-\beta J_0\halpha,\\
  R_2&\leq I(T\wedge V_2\vert V_1U)-(1-\alpha) I(Z\wedge V_2\vert V_1U)+C_2-(1-\beta)J_0\halpha,\\
  R_1+R_2&\leq I(T\wedge V_1V_2\vert U)-\alpha I(Z\wedge V_1\vert V_2U)-(1-\alpha)I(Z\wedge V_2\vert V_1U)
  \\&\hspace{.6\columnwidth}+C_1+C_2-J_0\halpha,\\
  R_1+R_2&\leq I(T\wedge V_1V_2)-I(Z\wedge V_1V_2).
\end{align*}

\begin{lemma}\label{case2confverein}
  We have for every $\alpha\in[\alpha\2_0,\alpha\2_1]$ 
\[
  \R\2_\alpha(p,C_1,C_2)=\bigcup_{\beta\zhalpha_0\leq\beta\leq\beta\zhalpha_1}\R\2_{\alpha,\beta}(p,C_1,C_2).
\]

\end{lemma}

This is seen immediately using Lemma \ref{gemconc}.

\subsubsection{For Case 3:}

Define
\[
  \beta\1_0:=[1-\frac{C_2}{I(Z\wedge V_1V_2)}]_+,\qquad\beta\1_1:=\min\{\frac{C_1}{I(Z\wedge V_1V_2)},1\}.
\]
We have $\beta\1_0\leq\beta\1_1$ because $I(Z\wedge V_1V_2)<C_1+C_2$.

\begin{lemma}
  For $\beta\3_0\leq\beta\leq\beta\3_1$, let $\R\3_\beta(p,C_1,C_2)$ be the set of those real pairs $(R_1,R_2)$ satisfying
\begin{align*}
  R_1&\leq I(T\wedge V_1\vert V_2U_0)+C_1-\beta I(Z\wedge V_1V_2),\\
  R_2&\leq I(T\wedge V_2\vert V_1U_0)+C_2-(1-\beta) I(Z\wedge V_1V_2),\\
  R_1+R_2&\leq\min\bigl\{I(T\wedge V_1V_2\vert U)+C_1+C_2-I(Z\wedge V_1V_2),\\
  &\qquad\qquad I(T\wedge V_1V_2)-I(Z\wedge V_1V_2)\bigr\}.
\end{align*}
Then
\[
  \R\1(p,C_1,C_2)=\bigcup_{\beta\1_0\leq\beta\leq\beta\1_1}\R\1_\beta(p,C_1,C_2).
\]
\end{lemma}

Thus it is sufficient to show the achievability of $\R\3_\beta(p,C_1,C_2)$ for every $\beta$. For the proof one uses Lemma \ref{gemconc}.

\subsection{Coding}\label{subsect:confcod}

Let $C_1,C_2>0$ and let $p\in\Pi_{C_1+C_2}$. Further let $(R_1,R_2)\in\R(p,C_1,C_2)$. In Case 1 we then know that there is a $\beta\in[\beta\1_0,\beta\1_1]$ such that $(R_1,R_2)\in\R\1_\beta(p,C_1,C_2)$, in Case 2 we have an $\alpha\in[\alpha\2_0,\alpha\2_1]$ and a $\beta\in[\beta\zhalpha_0,\beta\zhalpha_1]$ with $(R_1,R_2)\in\R\2_{\alpha,\beta}(p,C_1,C_2)$. For Case 3, there is a $\beta\in[\beta\3_0,\beta\3_1]$ with $(R_1,R_2)\in\R\3_\beta(p,C_1,C_2)$. Recall the notation 
\[
  J_0\halpha=\begin{cases}
      I(Z\wedge U)&\text{in Case 1},\\
      \alpha I(Z\wedge V_2U)+(1-\alpha)I(Z\wedge V_1U)&\text{in Case 2},\\
      I(Z\wedge V_1V_2)&\text{in Case 3}.
    \end{cases}
\]
 We set
\[
  \tilde R_0^{(1)}:=R_1\wedge(C_1-\beta J_0\halpha),\quad\tilde R_0^{(2)}:=R_2\wedge(C_2-(1-\beta)J_0\halpha)
\]
and
\[
  \tilde R_\nu:=R_\nu-\tilde R_0^{(\nu)}\qquad(\nu=1,2).
\]
Then setting 
\[
  \tilde R_0:=\tilde R_0^{(1)}+\tilde R_0^{(2)},
\]
we conclude that 
\[
  (\tilde R_0,\tilde R_1,\tilde R_2)\in\begin{cases}
                                         \R\1_{\beta}(p)&\text{in Case 1},\\
                                         \R\2_{\alpha,\beta}(p)&\text{in Case 2},\\
                                         \R\3_\beta(p)&\text{in Case 3}.
                                       \end{cases}
\]
In particular, $(\tilde R_0,\tilde R_1,\tilde R_2)$ is achievable by the wiretap MAC $W$ with common message under the common randomness bound $C_1+C_2$. That means that for any $\eta,\eps>0$ and for sufficiently large $n$, there is a common-message blocklength-$n$ code which has the form
\begin{align*}
  \tilde G&:[\tilde K_0]\times[\tilde K_1]\times[\tilde K_2]\rightarrow\P(\X^n\times\Y^n),\\
  \phi&:\T^n\rightarrow[\tilde K_0]\times[\tilde K_1]\times[\tilde K_2],
\end{align*}
and the proof of Theorem \ref{thmcomm} shows that we may assume that $\tilde G$ is given by
\[
  \tilde G(\x,\y\vert\tilde k_0,\tilde k_1,\tilde k_2)=\frac{1}{\tilde L_0}\sum_{l_0=1}^{\tilde L_0}\tilde G_1(\x\vert\tilde k_0,\tilde k_1,l_0)\tilde G_2(\y\vert\tilde k_0,\tilde k_2,l_0)
\]
for two stochastic matrices $\tilde G_1,\tilde G_2$. For $\tilde L_0$ we have the bounds
\[
  J_0\halpha+\frac{\eta}{4}\leq\frac{1}{n}\log\tilde  L_0\leq J_0\halpha+\frac{\eta}{2}.
\]
Without loss of generality we may additionally assume that $\tilde L_0\1:=\tilde L_0^\beta$ and $\tilde L_0\2:=\tilde L_0^{(1-\beta)}$ are  integers. If $0<2\eta<\min\{\tilde R_\nu:\nu=0,1,2,\tilde R_\nu>0\}$, the codelength triple $(\tilde K_0,\tilde K_1,\tilde K_2)$ may be assumed to satisfy
\begin{equation}\label{kbed}
  [\tilde R_\nu-2\eta]_+\leq\frac{1}{n}\log\tilde K_\nu\leq[\tilde R_\nu-\eta]_+,\qquad(\nu=0,1,2),
\end{equation}
and both the average error as well as $I(\tilde M_0\tilde M_1\tilde M_2\wedge Z^n)$ are upper-bounded by $\eps$, where $(\tilde M_0,\tilde M_1,\tilde M_2)$ is distributed uniformly on $[\tilde K_0]\times[\tilde K_1]\times[\tilde K_2]$ and $Z^n$ is Eve's corresponding output random variable. The definitions imply that 
\[
  \frac{1}{n}\log\tilde K_0\tilde L_0\leq C_1+C_2.
\]
We can find $\tilde K_0',\tilde K_0^{(1)},\tilde K_0^{(2)}$ such that $\tilde K_0'=\tilde K_0^{(1)}\tilde K_0^{(2)}$ and $\tilde K_0'\leq\tilde K_0$ and satisfying
\begin{align}
  [\tilde R_0^{(\nu)}-2\eta]_+\leq\frac{1}{n}\log \tilde K_0^{(\nu)}\leq[\tilde R_0^{(\nu)}-\frac{\eta}{2}]_+,\label{k1bed}\\
  [\tilde R_0-2\eta]_+\leq\frac{1}{n}\log \tilde K_0'.\label{k2bed}  
\end{align}
Thus one obtains a natural embedding
\begin{align}\label{embed}
  [\tilde K_0^{(\nu)}]\times[L_0^{(\nu)}]\subset [\lfloor2^{nC_\nu}\rfloor]&&(\nu=1,2).
\end{align}

We now construct a wiretap code with conferencing encoders. Let 
\begin{align*}
  K_\nu:=\tilde K_0^{(\nu)}\tilde K_\nu&&(\nu=1,2).
\end{align*}
Thus every $k_\nu\in[K_\nu]$ has the form $(a_\nu(k_{\nu}),b_\nu(k_{\nu}))$ with $a_\nu(k_\nu)\in [\tilde K_0^{(\nu)}]$ and $b_\nu(k_\nu)\in[\tilde K_\nu]$. We then define a stochastic one-shot Willems conferencing protocol
\[
  c_1:[K_1]\rightarrow\P([\lfloor 2^{nC_1}\rfloor]),\qquad c_2:[K_2]\rightarrow\P([\lfloor 2^{nC_2}\rfloor])
\]
which is used to generate both a common message as well as common randomness. Given a message $k_\nu\in[K_\nu]$, Alice\nue\ chooses an $l_\nu$ uniformly at random from the set $[L_0^{(\nu)}]$ and then maps the pair $(k_\nu,l_\nu)$ to $(a_\nu(k_\nu),l_\nu)$, so $c_\nu(k_\nu,l_\nu)=(a_\nu(k_\nu),l_\nu)$.

Next we define stochastic encoders $G_1,G_2$ as in the definition of a code with conferencing encoders by setting 
\[
  \J:=[\lfloor 2^{nC_1}\rfloor]\times[\lfloor 2^{nC_2}\rfloor]
\]
and, using the embedding \eqref{embed},
\[
  G_1(\x\vert k_1,j)=\tilde G_1(\x\vert(a_1(k_1),k_0^{(2)}),b_1(k_1),(l_1,l_2))
\]
if $j=((a_1(k_1),l_1),(k_0^{(2)},l_2))$ and letting $G_1(\x\vert k_1,j)$ be arbitrary else; $G_2$ is defined analogously. For decoding, one takes the decoder from the common message code and lets it combine the messages it receives into elements of $[K_1]$ and $[K_2]$. By \eqref{kbed}, the numbers $K_1$ and $K_2$ satisfy
\begin{align*}
  \frac{1}{n}\log K_1&\geq R_1-3\eta,\\
  \frac{1}{n}\log K_2&\geq R_2-3\eta.
\end{align*}
Thus depending on the case we are in, every rate pair $(R_1,R_2)$ contained in $\R\1_\beta(p,C_1,C_2)$ or $\R\2_{\alpha,\beta}(p,C_1,C_2)$ or $\R\3_\beta(p,C_1,C_2)$ is achievable.

\section{Discussion}\label{sect:discussion}

\subsection{Conferencing and Secret Transmission}\label{subsect:confsectrans}

This subsection is devoted to the comparison of the wiretap MAC without conferencing nor common randomness and the wiretap MAC if conferencing is allowed. As our focus is on conferencing, we assume that common randomness can only be established by conferencing. We show that there exists a wiretap MAC where the only rate pair contained in the region \eqref{zero-achi} achievable without conferencing is $(0,0)$, whereas if conferencing is enabled with arbitrarily small $C_1,C_2>0$, then the corresponding achievable region contains positive rates. Note that this does not mean that there are cases where conferencing is necessary to establish secret transmission as we do not have a converse. This restriction limits the use of this discussion and should be kept in mind.

Our goal is to find multiple access channels $W_b$ and $W_e$ such that for every Markov chain $((V_1,V_2),(X,Y),(T,Z))$ where $P_{T\vert XY}=W_b$ and $P_{Z\vert XY}=W_e$ and where $V_1$ and $V_2$ are independent one has
\begin{equation}\label{doubleconv}
  I(T\wedge V_1V_2)\leq I(Z\wedge V_1V_2).
\end{equation}
We noted in Remark \ref{remzeroachi} that \eqref{zero-achi} is the achievable region without conferencing and it is easy to see that condition \eqref{doubleconv} is an equivalent condition for this region to equal $\{(0,0)\}$. Thus the only rate pair which is achievable according to our above considerations is $(R_1,R_2)=(0,0)$. At the same time, there should be a Markov chain $(U,(X,Y),(T,Z))$ for the same pair of channels $W_b$ and $W_e$ such that 
\[
  I(T\wedge XY)>I(Z\wedge XY).
\]
This would prove the existence of a rate pair $(R_1,R_2)$ with positive components for arbitrary $C_1,C_2>0$.

We recall one concept of comparison for single-sender discrete memoryless channels (DMCs) introduced by K\"orner and Marton \cite{KMComp}. 

\begin{definition}
  A DMC $W_e:\X\rightarrow\P(\Z)$ is less noisy than a DMC $W_b:\X\rightarrow\P(\T)$ if for every Markov chain $(U,X,(T,Z))$ with $P_{T\vert X}=W_b$ and $P_{Z\vert X}=W_e$ one has
\[
  I(Z\wedge U)\geq I(T\wedge U).
\]
\end{definition}
 It was observed by van Dijk \cite{vD} that this is nothing but saying that the function
\begin{align*}
  P_X\mapsto I(Z\wedge X)-I(T\wedge X),&&P_X\in\P(\X)
\end{align*}
is concave. Now we generalize this to the MAC case to obtain an equivalent condition for \eqref{doubleconv}.

\begin{lemma}\label{vDMAC}
  \eqref{doubleconv} holds for every Markov chain $((V_1,V_2),(X,Y),(T,Z))$ with independent $V_1,V_2$ and $X$ independent of $V_2$ and $Y$ independent of $V_1$ and $P_{T\vert XY}=W_b$ and $P_{Z\vert XY}=W_e$ if and only if the function
\begin{align*}
  (P_X,P_Y)\mapsto I(Z\wedge XY)-I(T\wedge XY),&&X,Y\text{ independent r.v.s on }\X\times\Y
\end{align*}
is concave in each of its components.
\end{lemma}
\begin{proof}
  Let a Markov chain be given as required in the lemma. One has
\begin{align}
  &\mathrel{\hphantom{=}}I(Z\wedge V_1V_2)-I(T\wedge V_1V_2)\label{nonneg}\\
  &=\bigl(I(Z\wedge XY)-I(T\wedge XY)\bigr)-\bigl(I(Z\wedge XY\vert V_1V_2)-I(T\wedge XY\vert V_1V_2)\bigr).\notag
\end{align}
Now note that the rightmost bracket equals
\begin{multline*}
  \sum_{v_1}\sum_{v_2}P_{V_1}(v_1)P_{V_2}(v_2)\\\bigl(I(Z\wedge XY\vert V_1=v_1,V_2=v_2)-I(T\wedge XY\vert V_1=v_1,V_2=v_2)\bigr),
\end{multline*}
so it is clear that the nonnegativity of \eqref{nonneg} is equivalent to the concavity in each component of the function from the lemma statement.\qed
\end{proof}

We now define the channels $W_b$ and $W_e$ which will provide the desired example. Let $N_1,N_2$ be i.i.d.\ random variables uniformly distributed on $\{0,1\}$. The input alphabets are $\X=\Y=\{0,1\}$. The output alphabet of $W_b$ is $GF(3)$ and the output alphabet of $W_e$ is $\{-2,\ldots,3\}$. The outputs $t$ of $W_b$ are given by
\[
  t = x+y+N_1,
\]
those of $W_e$ by
\[
  z = 2x-2y+N_2.
\]
The intuition is that in $W_e$, one can exactly determine through the output whether or not the inputs were equal and if they were unequal, which input was 0 and which was 1. For $W_b$, however, there are for every output at least two input possibilities, so it is reasonable that an independent choice of the inputs makes $W_e$ better than $W_b$. However, if one may choose the inputs with some correlation, one may choose the inputs to be equal. Then the output of $W_e$ is only noise, whereas one can still extract some information about the input from $W_b$.

As the entries of the corresponding stochastic matrices of both channels are only $1/2$ or $0$, the conditional output entropy is independent of the input distribution and equals $1$. Further any pair of independent random variables on $\X$ and $\Y$ is given by parameters $q,r\in[0,1]$ such that
\[
  \PP[X\hq=0]=q,\qquad\PP[Y\hr=0]=r.
\]
Thus in order to determine whether \eqref{doubleconv} holds, it is enough to consider the function $H(Z\hqr)-H(T\hqr)$ for $T\hqr,Z\hqr$ being the outputs of $W_b$ and $W_e$, respectively, corresponding to the pair $(X\hq,Y\hr)$. One has
\begin{align*}
  f_Z(q,r):=H(Z\hqr)&=-q(1-r)\log(q(1-r)/2)\\
  &\mathrel{\hphantom{=}}-(qr+(1-q)(1-r))\log((qr+(1-q)(1-r))/2)\\
  &\mathrel{\hphantom{=}}-(1-q)r\log((1-q)r/2)
\end{align*}
and
\begin{align*}
  f_T(q,r)&:=H(T\hqr)\\
  &=-\frac{1}{2}(qr+(1-q)(1-r))\log((qr+(1-q)(1-r))/2)\\
  &\mathrel{\hphantom{=}}-\frac{1}{2}(qr+q(1-r)+(1-q)r)\log((qr+q(1-r)+(1-q)r)/2)\\
  &\mathrel{\hphantom{=}}-\frac{1}{2}(q(1-r)+(1-q)r+(1-q)(1-r))\cdot\\
  &\hspace{3cm}\cdot\log((q(1-r)+(1-q)r+(1-q)(1-r))/2).
\end{align*}
Both entropies are symmetric in $q$ and $r$ and continuous on $[0,1]^2$ and differentiable on $(0,1)^2$, so by Lemma \ref{vDMAC} it suffices to find the second derivatives in $q$ of both of them and to compare. 

We have
\begin{align*}
  \frac{\partial f_Z}{\partial q}(q,r)=&-(1-r)\log(q(1-r)/2)\\
  &-(2r-1)\log((qr+(1-q)(1-r))/2)\\
  &+r\log((1-q)r/2)
\end{align*}
and 
\begin{align*}
  \frac{\partial f_T}{\partial q}(q,r)&=-\frac{1}{2}(2r-1)\log((qr+(1-q)(1-r))/2)\\
  &\mathrel{\hphantom{=}}-\frac{1}{2}(1-r)\log((qr+q(1-r)+(1-q)r)/2)\\
  &\mathrel{\hphantom{=}}+\frac{r}{2}\log((q(1-r)+(1-q)r+(1-q)(1-r))/2).
\end{align*}
Thus
\begin{align*}
  \frac{\partial^2f_Z}{\partial q^2}(q,r)=-\frac{1-r}{q}-\frac{(2r-1)^2}{qr+(1-q)(1-r)}-\frac{r}{1-q}
\end{align*}
and 
\begin{align*}
  \frac{\partial^2f_T}{\partial q^2}(q,r)&=-\frac{(2r-1)^2}{2(qr+(1-q)(1-r))}\\
  &\mathrel{\hphantom{=}}-\frac{(1-r)^2}{2(qr+q(1-r)+(1-q)r)}\\
  &\mathrel{\hphantom{=}}-\frac{r^2}{2(q(1-r)+(1-q)r+(1-q)(1-r))}.
\end{align*}
After some algebra, it turns out that for $q,r\in(0,1)$,
\begin{align*}
  \frac{\partial^2f_Z}{\partial q^2}(q,r)-\frac{\partial^2f_T}{\partial q^2}(q,r)&=-\frac{1-r}{2q}\cdot\frac{q+2r-qr}{q+r-qr}\\
  &\mathrel{\hphantom{=}}-\frac{(2r-1)^2}{2(qr+(1-q)(1-r))}\\
  &\mathrel{\hphantom{=}}-\frac{r}{2(1-q)}\cdot\frac{2-r-qr}{1-qr}\\
  &<0.
\end{align*}
Thus $f_Z-f_T$ is concave and \eqref{doubleconv} is true for $W_b,W_e$.

Now we show that there exists an input distribution with $I(T\wedge XY)>I(Z\wedge XY)$. Of course, $X$ and $Y$ cannot be independent any more in this case. Every probability distribution $p$ on $\{0,1\}$ induces a probability distribution $p^2$ on $\{0,1\}^2$ via $p^2(x,x)=p(x)$. Let the pair $(X,Y)$ be distributed according to $p$. It is immediate from the definition of $W_e$ that $I(Z\wedge XY)=0$. On the other hand, $P_T$ can be described by the vector $(1/2)(1,p(0),p(1))$. Thus one sees easily that this is maximized for $p(0)=p(1)=1/2$, resulting in 
\[
  I(T\wedge XY)=\frac{1}{2}.
\]
$p^2$ is identified as an element of $\Pi$ by setting $\U=\{0,1\}$, $P_U=P_X$, and $P_{X\vert U}=P_{Y\vert U}=\delta_U$. Note that $I(Z\wedge U)=0$, so secret transmission is possible with arbitrarily small conferencing capacities $C_1,C_2>0$.

\subsection{Necessity of Time-Sharing in Random Coding}\label{subsect:timesharing}

We show here that doing time-sharing during random coding is necessary for our method to work. This only serves to justify the effort we had to make in coding. We concentrate on Case 0 and 1. Then we have to show that it may happen that $\alpha\1_0>0$ or $\alpha\1_1<1$. Let $\X=\Y=\T=\Z=\{0,1\}$ and let $W_b,W_e:\{0,1\}^2\rightarrow\P(\{0,1\})$ be defined by 
\[
  W_b=\begin{pmatrix}
	0.6178& \quad0.3822\\
	0.0624& \quad0.9376\\
	0.9350& \quad0.0650\\
	0.2353& \quad0.7647
      \end{pmatrix}
  ,\qquad W_e=\begin{pmatrix}
                0.0729& \quad0.9271\\
		0.7264& \quad0.2736\\
		0.3662& \quad0.6338\\
		0.4643& \quad0.5357
              \end{pmatrix},
\]
where the output distribution for the input pair $(x,y)$ is given in row number $2x+y$ for each matrix. With $q=0.6933$ and $r=0.3151$, let $p=p^{(q)}\otimes p^{(r)}\in\P(\X\times\Y)$ be the product measure with the marginals
\[
  p^{(q)}=(q,1-q),\qquad p^{(r)}=(r,1-r).
\]
Note that $p\in\Pi_0$. One obtains the following entropies:
\begin{align*}
  H(T\vert XY)&\approx0.5685, & H(Z\vert XY)&\approx0.7851,\\
  H(T\vert X) &\approx0.8532, & H(Z\vert X) &\approx0.9952,\\
  H(T\vert Y) &\approx0.6251, & H(Z\vert Y) &\approx0.8442,\\
  H(T)        &\approx0.8866, & H(Z)        &\approx0.9999.
\end{align*}
Calculating with the above values returns
\begin{align*}
  I(T\wedge XY)      &=0.3181, & I(Z\wedge XY)      &=0.2147,\\
  I(T\wedge X\vert Y)&=0.0566, & I(Z\wedge X\vert Y)&=0.0590,\\
  I(T\wedge Y\vert X)&=0.2847, & I(Z\wedge Y\vert X)&=0.2101,\\
                     &         & I(Z\wedge X)       &=0.0047,\\
                     &         & I(Z\wedge Y)       &=0.1557.
\end{align*}
Thus the conditions \eqref{HC01} and \eqref{HC02} are satisfied. If $H_C<\min\{I(Z\wedge X\vert Y),I(Z\wedge Y\vert X)\}=0.0590$, then we can only show that $\R\0(p)$ or $\R\1(p)$ is achievable and might have to use time-sharing during random coding to do so. In fact, this is necessary as
\[
  I(Z\wedge X\vert Y)>I(T\wedge X\vert Y),
\]
whereas
\[
  I(Z\wedge Y\vert X)<I(T\wedge Y\vert X).
\]
Hence $\alpha\1_0>0$, but $\alpha\1_1=1$. This example was found by a brute-force search using the computer.




\appendix

\section{Proof of Lemma \ref{gemconc}}\label{sect:lemmaproof1}

The direction ``$\subset$'' in \eqref{statementalpha-union} is obvious. For the other direction, let $(R_0,R_1,R_2)\in\K$. We may assume that for some $0\leq\beta\leq1$, 
\begin{align*}
  R_1&=r_1-\beta(\alpha_1a_1+(1-\alpha_1)b_1)-(1-\beta)(\alpha_0a_1+(1-\alpha_0)b_1)\\
     &=r_1-(\beta\alpha_1+(1-\beta)\alpha_0)a_1-(\beta(1-\alpha_1)+(1-\beta)(1-\alpha_0))b_1
\end{align*}
because the claim is obvious for $R_1\leq r_1-\alpha_1a_1-(1-\alpha_1)b_1$. We show that $(R_0,R_1,R_2)\in\K_{\beta\alpha_1+(1-\beta)\alpha_0}$. The $R_1$-bound is satisfied due to our assumption. Further due to the bound on $R_1+R_2$,
\begin{align*}
  &\mathrel{\hphantom{\leq}}R_2\\
  &\leq r_{12}-c-r_1+(\beta\alpha_1+(1-\beta)\alpha_0)a_1+(\beta(1-\alpha_1)+(1-\beta)(1-\alpha_0))b_1\\
  &\leq r_{2}-(\beta\alpha_1+(1-\beta)\alpha_0)a_2-(\beta(1-\alpha_1)+(1-\beta)(1-\alpha_0))b_2,
\end{align*}
so $R_2$ also satisfies the necessary upper bound. The sum constraints are independent of $\alpha$. Hence all upper bounds in the definition of $\K_{\beta\alpha_1+(1-\beta)\alpha_0}$ are satisfied, and Lemma \ref{gemconc} is proved.

\section{Proof of Lemma \ref{unionconv2}}\label{sect:lemmaproof2}

For $\alpha\in[\alpha_0,\alpha_1]$, the set $\K_\alpha$ is contained in the convex hull of $\K_{\alpha_0}\cup\K_{\alpha_1}$. Thus we only have to prove that $\K=conv(\K_{\alpha_0}\cup\K_{\alpha_1})$. Without loss of generality we assume that $b>a$. 

We first prove $conv(\K_{\alpha_0}\cup\K_{\alpha_1})\subset\K$. Let $(R_0,R_1,R_2)\in conv(\K_{\alpha_0}\cup\K_{\alpha_1})$. Using the convexity of $\K_{\alpha_0}$ and $\K_{\alpha_1}$ we infer that there is a $(R\0_0,R\0_1,R\0_2)\in\K_{\alpha_0}$ and a $(R\1_0,R\1_1,R\1_2)\in\K_{\alpha_1}$ and a $\beta\in[0,1]$ such that
\[
  (R_0,R_1,R_2)=\beta(R\0_0,R\0_1,R\0_2)+(1-\beta)(R\1_0R\1_1,R\1_2).
\]
One sees immediately that $(R_0,R_1,R_2)$ satisfies the bounds \eqref{conv21}-\eqref{conv23} and \eqref{conv25}. It is sufficient to check that \eqref{conv24} is satisfied by the triples $(R\0_0,R\0_1,R\0_2)$ and $(R\1_0,R\1_1,R\1_2)$. For $(R\0_0,R\0_1,R\0_2)$ we assume that 
\[
  R\0_1=\gamma(r_1-\alpha_0a)
\]
for some $\gamma\in[0,1]$. After some calculations this yields
\begin{align*}
  bR\0_1+aR\0_2&\leq(b-a)r_1+ar_{12}-ab-(1-\gamma)(b-a)(r_1-\alpha_0a)\\
  &\leq (b-a)r_1+ar_{12}-ab.
\end{align*}
One proceeds analogously for $(R\1_0,R\1_1,R\1_2)$.

Next we have to check that $\K\subset conv(\K_{\alpha_0}\cup\K_{\alpha_1})$. It is sufficient to check whether those points $(R_0,R_1,R_2)$ are contained in $conv(\K_{\alpha_0}\cup\K_{\alpha_1})$ that satisfy both \eqref{conv24} and one of \eqref{conv21}-\eqref{conv23} with equality. So assume that
\begin{equation}\label{conv24erf}
  bR_1+aR_2=r_{12}a +r_1(b-a)-ab.
\end{equation}
First we also assume that
\[
  R_1+R_2=r_{12}-\alpha_0a-(1-\alpha_1)b.
\]
Then 
\[
  R_2=r_{12}-\alpha_0a-(1-\alpha_1)b-R_1
\]
and using \eqref{conv24erf} we obtain 
\[
  R_1=r_1-\frac{\alpha_1b-\alpha_0a}{b-a}a\leq r_1-\alpha_1a.
\]
For $R_2$ this gives 
\[
  R_2=r_{12}-r_1-\left(\alpha_0+\frac{\alpha_1b-\alpha_0a}{b-a}\right)a-(1-\alpha_1)b\leq r_2-(1-\alpha_1)b,
\]
so $(R_1,R_2)\in\K_{\alpha_1}$.

Now we assume
\[
  R_1=r_1-\alpha_0a.
\]
Then inserting this in \eqref{conv24erf} one obtains 
\[
  R_2\leq r_2-(1-\alpha_0)b,
\]
so $(R_1,R_2)\in\K_{\alpha_0}$. 

Finally for 
\[
  R_2=r_2-(1-\alpha_1)b
\]
we obtain 
\[
  R_1\leq r_1-\alpha_1a,
\]
so $(R_1,R_2)\in\K_{\alpha_1}$. This proves the lemma.

\end{document}